\documentclass[reprint,aps,twocolumn,notitlepage,pra,superscriptaddress,footinbib]{revtex4-2}

\usepackage[usenames,dvipsnames]{color}

\usepackage{physics}
\usepackage{amsmath}
\usepackage{amssymb}
\usepackage{amsthm}
\usepackage{bm}
\usepackage{bbm}
\usepackage{latexsym} 
\usepackage{mathtools}
\usepackage{mathrsfs}
\usepackage{tensor}
\usepackage{tcolorbox}
\usepackage{CJKutf8}
\usepackage{microtype}

\usepackage{thmtools}
\declaretheorem[
]{theorem}
\declaretheorem[
]{lemma}

\declaretheorem[
]{proposition}
\declaretheorem[
]{corollary}

\usepackage{CJKutf8}

\newcommand{\defint}[3]{\int_{#1}^{#2}{#3}}
\newcommand{\defsum}[3]{\sum_{#1}^{#2}{#3}}

\usepackage{graphicx} % for figures
\usepackage[caption=false]{subfig} % subfigures (via \subfloat[]{})
\usepackage[inline]{enumitem} % in-line lists and \setlist{} (below)
\setlist[enumerate,1]{label={(\roman*)}} % default to roman numbering
\setlist{nolistsep} % more compact spacing between environments
\usepackage{comment} % comment out large sections of text

\usepackage{multirow}
\usepackage{hhline}

\usepackage{natbib}
\definecolor{themecolor}{HTML}{d6791c}
\usepackage[colorlinks=true,allcolors=themecolor]{hyperref}
\usepackage[capitalize,noabbrev,nameinlink]{cleveref}
\crefname{appendix}{App.}{Apps.}
\crefname{definition}{Def.}{Defs.}
\crefname{equation}{Eq.}{Eqs.}
\crefname{figure}{Fig.}{Figs.}
\crefname{proposition}{Prop.}{Props.}
\crefname{table}{Tab.}{Tabs.}
\Crefname{table}{Table}{Tables}
\crefname{theorem}{Thm.}{Thms.}
\crefname{section}{Sec.}{Secs.}

\let\originalleft\left
\let\originalright\right
\renewcommand{\left}{\mathopen{}\mathclose\bgroup\originalleft}
\renewcommand{\right}{\aftergroup\egroup\originalright}

\begin{document}
\begin{CJK}{UTF8}{gbsn}
\renewcommand{\labelenumii}{\theenumii}
\renewcommand{\theenumii}{\theenumi.\arabic{enumii}.}

\newcommand{\Blair}{Montgomery Blair High School, Silver Spring, MD 20901, USA}
\newcommand{\QuICS}{Joint Center for Quantum Information and Computer Science, National Institute of Standards and Technology and University of Maryland, College Park, MD 20742, USA}
\newcommand{\JQI}{Joint Quantum Institute, National Institute of Standards and Technology and University of Maryland, College Park, MD 20742, USA}
\newcommand{\UMDPhysics}{Department of Physics, University of Maryland, College Park, MD 20742, USA}
\newcommand{\UMDCS}{Department of Computer Science and Institute for Advanced Computer Studies, University of Maryland, College Park, Maryland 20742, USA}

\newcommand{\thetitle}{Moving quantum sensors}
\title{\thetitle}
\author{Evan Zhang}
\affiliation{\Blair}
\affiliation{\QuICS}
\affiliation{\JQI}
\author{Erfan Abbasgholinejad}
\affiliation{\QuICS}
\affiliation{\JQI}
\affiliation{\UMDPhysics}
\author{\begin{CJK}{UTF8}{gbsn}Yu-Xin Wang (王语馨)\end{CJK}}
\affiliation{\QuICS}
\author{Ali~Fahimniya}
\affiliation{\QuICS}
\affiliation{\JQI}
\author{Daniel~J.~Spencer}
\affiliation{\QuICS}
\affiliation{\JQI}
\affiliation{\UMDPhysics}
\author{Sean R. Muleady}
\affiliation{\QuICS}
\affiliation{\JQI}
\author{Mohammad Hafezi}
\affiliation{\JQI}
\affiliation{\QuICS}
\author{Andrew~M.~Childs}
\affiliation{\QuICS}
\affiliation{\UMDCS}
\author{Alexey V. Gorshkov}
\affiliation{\QuICS}
\affiliation{\JQI}

\date{\today}
\begin{abstract}
    We investigate whether controlled motion of a quantum sensor offers a metrological advantage in a field-detection task. We consider a localized field of known strength $B$ occupying an unknown contiguous region making up a fraction $1/N$ of the total search space, with a qubit sensor controlled for a total time $T$ to maximize the probability of detecting the field. We study the regime in which the characteristic timescale of the field is slow compared to the time to move the sensor between any two regions, and we distinguish classical motion, where the sensor follows a definite trajectory, from quantum motion, where the sensor can be in a coherent superposition of locations. For a single classically moving sensor, we derive an optimal multi-stop protocol and obtain the exact maximum detection probability. Allowing coherent spatial superposition maps the single-sensor problem onto Hamiltonian search and reduces the interrogation time for successfully detecting the field from $\Theta(N/B)$ to $\Theta(\sqrt{N}/B)$. We then extend the problem to $M$ sensors and identify a hierarchy of optimal time scalings for constant success probability determined by the available spatial and inter-sensor quantum resources. Unentangled coherently moving sensors achieve $T=\Theta(\sqrt{N/M}/B)$ by searching disjoint regions in parallel, while entanglement among the sensors can further enhance the effective field coupling and achieve $T=\Theta(\sqrt{N}/(MB))$. Matching lower bounds establish the optimality of these scalings in the presence of arbitrary controls and ancilla systems. Our results separate the roles of classical motion, spatial coherence, and entanglement in spatial field detection.
\end{abstract}

\maketitle
\end{CJK}

\section{Introduction}\label{sec:intro}
Quantum sensors enable precision measurement of weak fields using coherent quantum dynamics~\cite{degen2017,Giovannetti2011AdvancesQuantumMetrology}. In standard metrological settings the sensor locations are fixed, and one optimizes over probe states, control pulses, and measurement choices for enhancement in different sensing tasks~\cite{boixo2007generalized,Giovannetti2006QuantumMetrology,degen2017,Giovannetti2011AdvancesQuantumMetrology,khan2025quantum}. A less studied complementary control knob is \emph{where} the sensors are placed in space and how---if at all---they are moved during the interrogation time~\cite{aigner2025power, eldredge2018,abbasgholinejad2025optimally}.

Experimentally, the controlled motion of qubits is rapidly becoming routine across several platforms. In solid-state sensing, single-spin magnetometers based on nitrogen vacancy centers can be integrated into scanning probes, enabling nanoscale spatial mapping by physically translating the sensor relative to a sample~\cite{Maletinsky2012ScanningNV,Schirhagl2014NVReview,Rondin2014NVReview}. In trapped-ion architectures, qubits can be shuttled between zones of segmented traps with high fidelity, motivated originally by scalable multi-zone designs and demonstrated in increasingly optimized transport protocols~\cite{Kielpinski2002IonTrapArchitecture,Blakestad2009XJunction,Bowler2012DiabaticTransport,Parazzoli2022Shuttling,ransford_98-qubit_2026}. In neutral-atom systems, coherent transport of single-atom qubits using optical tweezers has been demonstrated, and large-scale atom arrays rely on repeated rearrangement as a core primitive~\cite{Beugnon2007TweezerTransport,Barredo2016Assembler,Endres2016Assembly,BrowaeysLahaye2020RydbergReview, bluvstein2022quantum, salesrodriguez2025magic, constantinides2024optimal}.

Beyond classical transport, there have also been experimental realizations of coherent spatial delocalization of a qubit. For example, motional Schr{\"{o}}dinger-cat states of trapped ions create superpositions of spatially separated wave packets correlated with an internal qubit~\cite{Monroe1996CatState,saner2026generating, leindecker2026state, diaz2026arbitrary}. In optical lattices and tweezer arrays, state-dependent transport coherently splits and displaces atomic wave packets~\cite{Mandel2003SpinDependentTransport,greve_entanglement-enhanced_2022-1, aeppli_hamiltonian_2022, malia_distributed_2022, bluvstein2026fault, pampel2026motional, manetsch2025tweezer}. Related capabilities underlie experimental quantum walks, where an atom/ion coherently explores multiple positions and interferes between paths~\cite{Karski2009QuantumWalk,Schmitz2009IonQuantumWalk,young_tweezer-programmable_2022-1}.

These developments motivate a basic question: \emph{what is the role of sensor motion in metrology and is there any advantage associated with motional quantum coherence?}
To isolate the role of motion, we study a “detection” problem. A localized field (e.g., a magnetic field) of known magnitude $B$ occupies an unknown contiguous segment of the search space. When the sensor enters the field region, the field couples to the sensor's internal state and imprints a detectable phase. We aim to maximize the probability of detecting the field within a fixed total interrogation time $T$, and we ask how this performance depends on the allowed form of sensor motion. 
% whether allowing the sensor(s) to move improves this detection probability compared with stationary sensing strategies.

Crucially, we distinguish two notions of motion. In \emph{classical motion}, the sensor follows a classical trajectory and occupies a definite position at each time; see~\cref{fig:schematic}(a). In \emph{quantum motion}, the sensor’s spatial degree of freedom can remain in coherent superposition and be processed unitarily, enabling interference between different spatial locations; see~\cref{fig:schematic}(b). This distinction is analogous to the separation between classical and quantum search, and raises the question of whether quantum motion can similarly enhance the detection of the field~\cite{Grover1996STOC,Boyer1998TightBounds,Zalka1999Optimal,farhi1998analog,Childs2000PrecisionMeasurement}.

Our results characterize how these different kinds of motion affect the achievable detection performance. After establishing the optimal protocols for a single sensor with classical and quantum motion in \cref{sec:classical,sec:quantum}, we determine the optimal time scalings for $M$ moving sensors to reach constant success probability in \cref{sec:multi}. We show that coherent motion and entanglement among the sensors provide distinct advantages. Specifically, coherent motion yields a quadratic improvement from $N$ to $\sqrt{N}$, as in Grover search, while entanglement improves the dependence on the number of sensors from $1/\sqrt{M}$ to the Heisenberg scaling $1/M$.

\begin{figure}[tb]
    \centering
    \includegraphics[width=\linewidth]{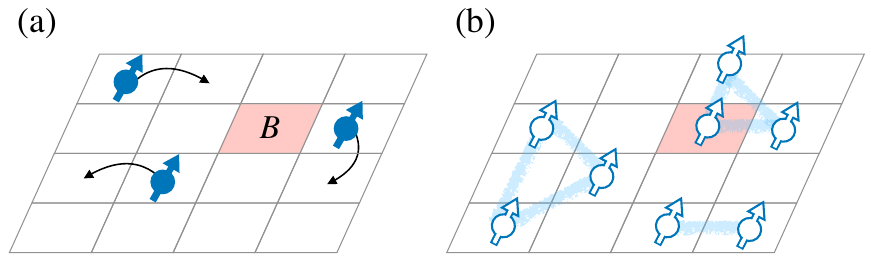}
    \caption{Schematic of the detection task. A field of known strength $B$ occupies an unknown localized region in the search space (red). A qubit sensor (blue) is controlled for a total interrogation time $T$ and then measured; the detection probability depends on the total time the sensor traverses the field region. (a) Classical motion: sensors occupy definite positions at all times and follow classical trajectories between sites. (b) Quantum motion: the sensor’s position degree of freedom can be prepared in a spatial superposition over multiple sites (represented by open arrows connected by thick translucent lines), enabling interference between different spatial locations.}
    \label{fig:schematic}
\end{figure}

\section{Model}\label{sec:setup}
We consider a one-dimensional ring of length $L$ containing a localized magnetic field of known magnitude $B$. The field is supported on the region $[x_0,x_0+\Delta]\text{ mod}(L)$, with an \emph{unknown} starting coordinate $x_{0}\in [0,L)$ with uniform prior distribution. We assume $N \coloneqq L/\Delta \in \mathbb{Z}$. A qubit sensor with an internal Hamiltonian coupling to the field is described by a position degree of freedom and an internal two-level system. When the sensor overlaps the field region, it evolves under a unitary generated by the Pauli-$Z$ operator, which we denote as $\sigma_{z}$. The corresponding Hamiltonian is then
\begin{equation}
    H_{x_{0}}(t) = B\sigma_{z}\otimes \int_{x_{0}}^{x_{0}+\Delta}\ketbra{x}\mathop{dx} + H_c(t),\label{eq:H}
\end{equation}
where $t \in [0,T]$. Here, $\ket{x}$ denotes a position eigenstate and $H_c(t)$ is a time-dependent control Hamiltonian acting on the sensor (including its internal and position Hilbert space) and any ancillary degrees of freedom. The control Hamiltonian is independent of the unknown position $x_0$ and is responsible for the internal dynamics, the motion of the sensor, and entanglement between the sensor and the ancilla. We assume these controls are arbitrarily fast, and in particular, neglect the time taken to move the sensor. This approximation is motivated by recent experiments, where qubit transport can occur on timescales much shorter than the qubit coherence time~\cite{bluvstein2022quantum,salesrodriguez2025magic,ransford_98-qubit_2026}. Thus, for weak fields with $1/B$ approaching the available coherence time, neglecting the transport time is a reasonable idealization.
This differs from locality-constrained search models~\cite{benioff2002space,aaronson2005quantum,childs2004spatial}, which consider a related problem where motion between neighboring sites has comparable cost to interrogating a given site. Moreover, this framework captures adaptive protocols through $H_c(t)$ and coupling the sensor to arbitrary ancillas.

We now consider the task of field detection. The system is initialized at time $t = 0$ and evolves for a total interrogation time $T$ under the Hamiltonian in~\cref{eq:H}. We denote the state at time $T$ by $\ket{\psi_{x_0}(T)}$, which depends on the position of the field $x_0$. In the absence of the field, the system evolves only under $H_c(t)$, yielding the final state $\ket{\psi_{\emptyset}(T)}$. In this case, the probability of successfully detecting the field at $x_0$ is $p_{x_0}(T) := 1 - \abs{\braket{\psi_\emptyset(T)}{\psi_{x_0}(T)}}^2$, and since $x_0$ is uniformly random, the average success probability is given by
\begin{equation}
    P(T) = 1-\frac{1}{L}\int_0^L \abs{\braket{\psi_{\emptyset}(T)}{\psi_{x_0}(T)}}^2\mathop{dx_0}.\label{eq:P_succ}
\end{equation}
We emphasize that this is a one-sided state-discrimination probability and a positive measurement outcome certifies the presence of the field. Its relation to other standard metrics is discussed in~\cref{app:state_discrimination}.

Our goal is to optimize the control term $H_c(t)$ in order to maximize $P(T)$ for a total time $T$, or equivalently to minimize the time $T$ to reach a target success probability. We compare three increasingly powerful models: (i) \emph{classical motion} of a single sensor, where the sensor follows an arbitrary classical trajectory $x(t)$ for $t \in [0,T]$; (ii) \emph{coherent (quantum) motion} of a single sensor, where the spatial degree of freedom is a quantum register that can remain in superposition and be acted on coherently; and (iii) $M$ sensors, where we consider both classical and quantum motion, and  entanglement between the sensors.

\section{Optimal protocol for classical motion}\label{sec:classical}
We first consider a \emph{single} sensor whose motion is classical, i.e., the sensor follows a trajectory $x(t)$ with a definite position at each time. We assume no ancilla system and restrict to single-shot field detection (i.e., no adaptive protocols). In this case, the optimal internal initial state of the sensor is $\ket{+} = (\ket{0}+\ket{1})/\sqrt{2}$, and since the interaction in~\cref{eq:H} depends only on whether the sensor lies inside the (unknown) field region, any classical protocol is fully characterized by the density of \emph{occupation time} $\tau(x)$ it assigns to each location. In particular, it suffices to specify a non-negative function $\tau(x)\ge 0$ for $x\in[0,L]$, with $\int_0^L \tau(x) dx = T$. Hence, the probability to detect the field in~\cref{eq:P_succ} is
\begin{equation}
    P_\mathrm{cl}(T) = \frac{1}{L}\int_0^L \sin^2\left(B\int_{x_0}^{x_0+\Delta}\tau(x)\mathop{dx}\right)\mathop{dx_0}.\label{eq:Pclassical_general}
\end{equation}
The optimization problem therefore reduces to maximizing~\cref{eq:Pclassical_general} subject to the constraints on the function $\tau(x)$.

\begin{theorem}[Optimal classical protocol]\label{thm:classical_opt}
    For $0 \le T \le \pi N/{(2B)}$, the maximum success probability over all classically moving single-sensor protocols is
    \begin{equation}
        P_\mathrm{cl}^\star(T) = \max_{1 \le k \le N}\;\frac{k}{N}\,\sin^2\left(\frac{BT}{k}\right).\label{eq:Pcl_opt}
    \end{equation}
    An optimal protocol achieving~\cref{eq:Pcl_opt} is a \emph{$k$-stop protocol}: choose $k$ distinct sites separated by $\Delta$ and dwell for an equal time $T/k$ at each point. Perfect detection, $P_\mathrm{cl}^\star(T)=1$, is achievable if and only if $T \ge \pi N/{(2B)}$.
\end{theorem}

\noindent\emph{Proof sketch.}
To prove the result, we first consider a discrete version of the problem. We divide the ring into $N=L/\Delta$ equally spaced sites, where the field occupies exactly one site. A classical sensing protocol is then specified by times $\{\tau_j\}_{j=1}^{N}$, where $\tau_j$ is the total time spent by the sensor at site $j$, with $\tau_j\geq0$ and $\sum_j\tau_j=T$. If the field is located at site $j$, the sensor interacts with it for time $\tau_j$. Thus, the probability in~\cref{eq:Pclassical_general} becomes
% We begin with the discretized problem on $N = L/\Delta$ equally spaced sites, where the field occupies exactly one site. 
% Specifying times $\{\tau_j\}_{j=1}^{N}$ for each point, the probability in~\cref{eq:Pclassical_general} 
$P_\mathrm{cl}(T) = \sum_{j=1}^{N}\sin^2(B \tau_j)/N$. For $T \le \pi N/{(2B)}$, we can restrict to $\tau_j \in[0,\pi/(2B)]$, because any time beyond $\pi/(2B)$ can be reduced without decreasing $\sin^2(Bt_j)$ and redistributed to increase the average. Next, optimizing $P_\mathrm{cl}(T)$ over $\{\tau_j\}_{j=1}^{N}$ subject to $\tau_j \ge 0$ and $\sum_j \tau_j = T$ yields~\cref{eq:Pcl_opt} (see~\cref{app:optimal_classical} for a more detailed derivation).

To pass from this discretization to the continuum functional optimization in~\cref{eq:Pclassical_general}, for any integer $r \ge 1$, we discretize the ring into $rN$ equally spaced sites so that the field spans $r$ consecutive sites. Grouping field starting positions by residue class modulo $r$ reduces the refined problem to $r$ independent copies of the $N$-site problem above, implying that the same optimal value from~\cref{eq:Pcl_opt} holds for all $rN$. Taking $r \to \infty$ then yields the continuous description in~\cref{eq:Pclassical_general}. A complete proof is provided in~\cref{app:optimal_classical}.

The optimal $k$-stop protocol achieves its performance by \emph{sudden jumps} between $\Delta$-separated locations, where we assume the travel time between stops is negligible compared to the interrogation time. In~\cref{app:const_velocity}, we study how moving a sensor at a constant velocity performs (i.e., no sudden jumps), and show that it is generally suboptimal at finite $N$ (see also~\cref{fig:classical_compare}), but its performance converges to the $k$-stop protocol as $N \to \infty$. In this large-$N$ limit, with $BT/N$ held constant, one may treat $k/N$ as a continuous variable. Optimizing~\cref{eq:Pcl_opt} yields $k^* \simeq 0.8580 BT$, and $P_\mathrm{cl}^\star(T) \simeq 0.7246 BT/N$. Since the sensor spends an equal time at each site before moving a distance $\Delta$, the  corresponding optimal velocity is $v^*\simeq \Delta k^*/T\simeq 0.8580 B\Delta$.

\begin{figure}[tb]
    \centering
    \includegraphics[width=\linewidth]{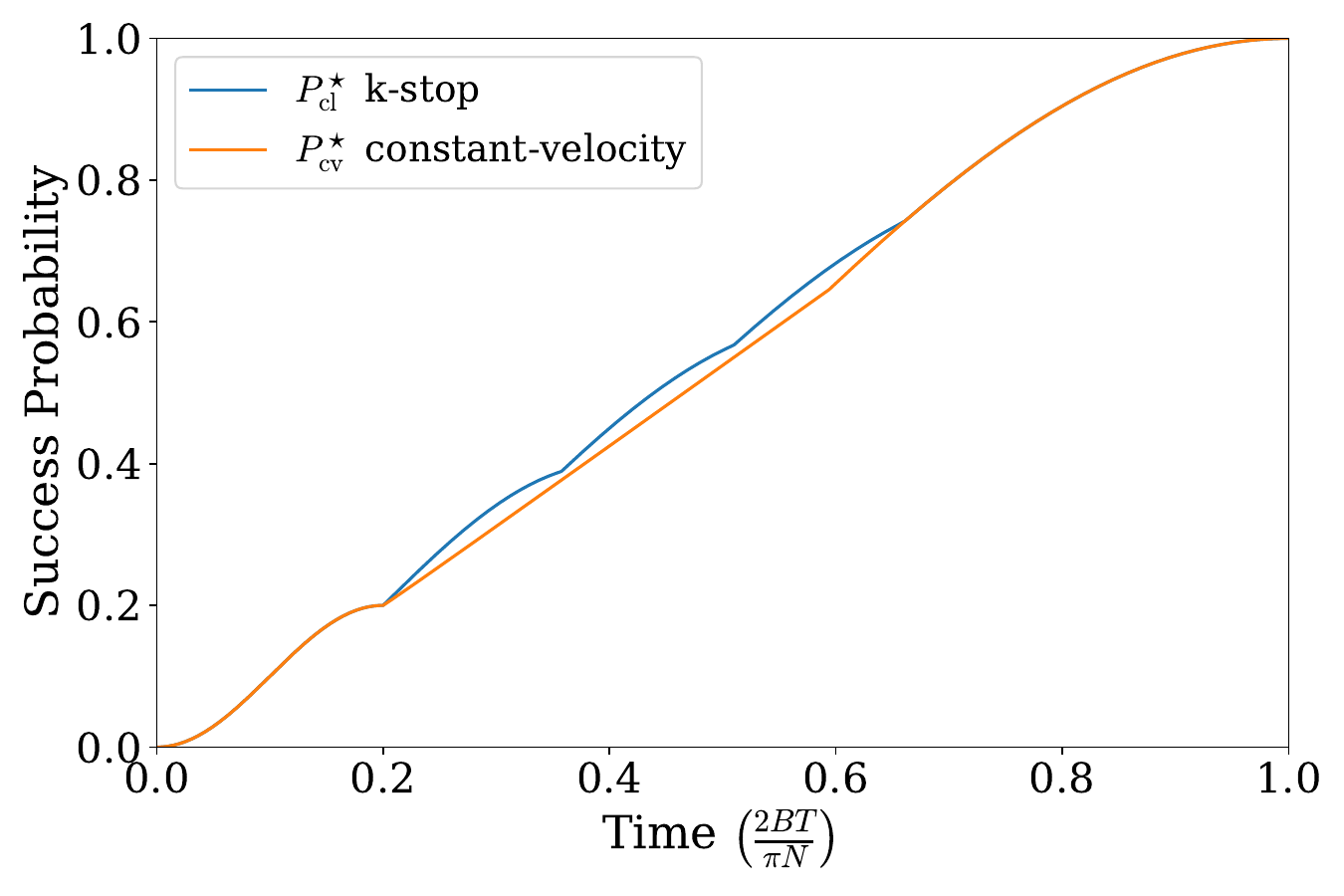}
    \caption{Classical motion protocols for $N=5$. Success probability versus time $T$ for the optimal $k$-stop protocol and the constant-velocity protocol.
    At short times both curves coincide with the stationary strategy ($v=0$, $k=1$), while at long times
    they approach the uniform-coverage limit ($k=N$, one full loop). In the intermediate regime the
    constant-velocity lies below the $k$-stop protocol; this gap shrinks as $N$ increases.}
    \label{fig:classical_compare}
\end{figure}

\section{Coherent motion and Grover-type enhancement}\label{sec:quantum}
We next discuss how allowing coherent spatial motion maps our sensing task to the setting of Refs.~\cite{farhi1998analog, Childs2000PrecisionMeasurement}, where the authors show that Grover search can be viewed as improving distinguishability among a set of Hamiltonians~\cite{Childs2000PrecisionMeasurement,allen2025quantum}. Similarly, we can use the unknown field region in our model as an oracle to achieve a $\Theta(\sqrt{N})$ search speed-up~\cite{Grover1996STOC, Childs2000PrecisionMeasurement, farhi1998analog, Zalka1999Optimal, bennett1997strengths}.

While our results are general, for ease of presentation, we 
work with the discrete model throughout the rest of this paper. All the lower bounds discussed in this Section hold for continuous $x$, as long as we replace $1/N$ by $\Delta /L$ (see App.~\ref{app: grover_proof} for detail). The position space of the sensor thus consists of $N = L/\Delta$ equally spaced sites with orthonormal position basis
$\{\ket{j}\}_{j=1}^{N}$, where the (unknown) field location corresponds to a single marked site $x_0 \in \{1,\dots,N\}$. We review two natural ways to exploit the oracle to localize the sensor at the unknown field location, starting from the uniform superposition state $\ket{s} = \frac{1}{\sqrt{N}}\sum_{j=1}^{N}\ket{j}$~\cite{Grover1996STOC, Childs2000PrecisionMeasurement,farhi1998analog}.

\paragraph{Grover iterations (oracle + reflection).}
Evolving under the Hamiltonian in~\cref{eq:H} with no control for time $\pi/B$ implements a phase oracle on the marked site while leaving the internal state unchanged:
\begin{equation}
    O_{x_0}
    \coloneqq \exp\!\left[-i\frac{\pi}{B}\, \big(B\,\sigma_z\otimes \ket{x_0}\!\bra{x_0}\big)\right] = I-2\,\ketbra{x_0}.\label{eq:oracle_def}
\end{equation}
Thus, the unknown field applies a $-1$ phase to $\ket{x_0}$ and acts trivially on $\ket{j\neq x_0}$, which is essentially the Grover oracle. Therefore, implementing the Grover iterations \cite{Grover1996STOC}, we can localize the sensor at the location of the field in time $\Theta({\sqrt{N}}/{B})$.

\paragraph{Continuous-time coherent-drive.}
Alternatively, we can add a controlled driving term $H_c = B\ketbra{s}$ while keeping the internal qubit in a $\sigma_z$ eigenstate so that the unknown field acts as a spatial potential $B\ketbra{x_0}$ on the position register. The effective Hamiltonian,
\begin{equation}
    H_{x_0} = B\left(\ketbra{x_0}+\ketbra{s}\right),\label{eq:Heff_drive}
\end{equation}
induces Rabi oscillations in the two-dimensional subspace spanned by $\ket{x_0}$ and the component of $\ket{s}$ orthogonal to $\ket{x_0}$, transferring amplitude from $\ket{s}$ to $\ket{x_0}$ in time $\Theta(\sqrt{N}/B)$ \cite{farhi1998analog}. 

After the localization step, we switch the internal state to $\ket{+}$ and apply an additional interaction interval of duration $\pi/(2B)$ to map $\ket{+} \mapsto \ket{-}$ on the sensor, yielding the same $\Theta(\sqrt{N}/B)$ time for detecting the field with constant success probability.

Moreover, the optimality of the $\Theta(\sqrt{N}/B)$ scaling follows directly from the optimality of Grover search and its continuous-time analogue \cite{bennett1997strengths,Zalka1999Optimal,farhi1998analog}.
Thus, we obtain the following result:

\begin{corollary}[Optimal quantum protocol]\label{cor:quantum_single}
    There exists a coherent-motion protocol achieving $P \ge 2/3$ in time $T = \mathcal{O}(\sqrt{N}/B)$. Furthermore, any coherent protocol with a single sensor obeys the lower bound $T = \Omega(\sqrt{N}/B)$ to achieve $P \ge 2/3$.
\end{corollary}

We emphasize that this speed-up has already been discussed in Refs.~\cite{farhi1998analog,Childs2000PrecisionMeasurement} for Hamiltonian discrimination, and our observation is to identify coherent spatial motion as the physical resource that realizes the Hamiltonian oracle in the present sensing problem. It also allows a direct comparison between the known quantum-search scaling and the exact classical moving-sensor optimum derived in~\cref{sec:classical}. In~\cref{fig:quantum_vs_classical}, we highlight the separation between quantum and classical motion, in which the Grover-style protocol achieves high detection probability at much shorter times than the optimal classical strategy.

\begin{figure}[tb]
    \centering
    \includegraphics[width=\linewidth]{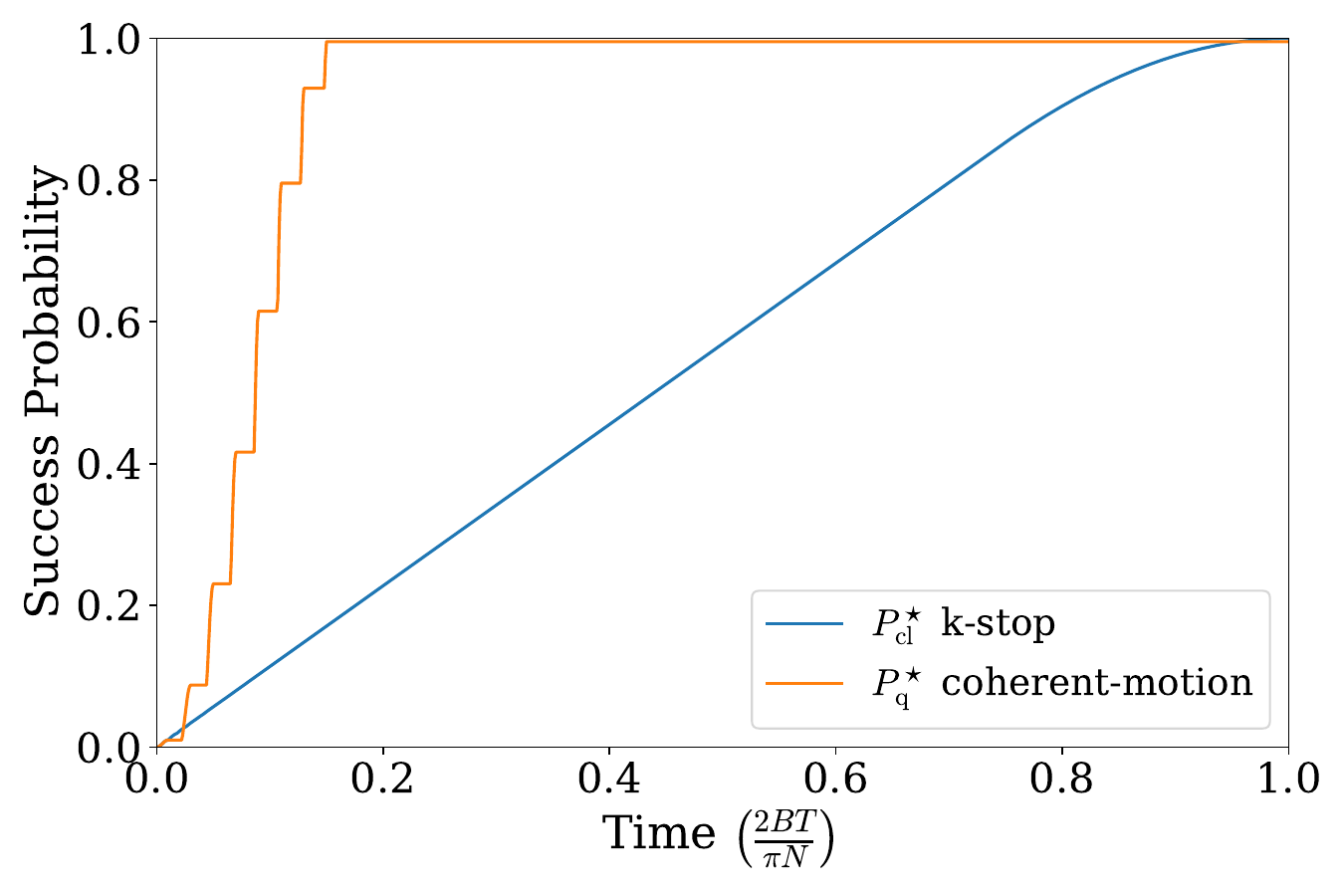}
    \caption{Optimal success probability versus interrogation time $T$ for the optimal classical protocol $P_\mathrm{cl}^\star(T)$ (the $k$-stop protocol) and the coherent-motion protocol $P_\mathrm{q}^\star(T)$ with $N = 100$. The quantum protocol reaches high success probability at shorter times, exhibiting the expected $\Theta(\sqrt{N})$ speedup. In this large-$N$ limit the classical optimum shows linear scaling $P_\mathrm{cl}^\star(T)\propto BT/N$, consistent with the continuous-$k$ optimization of the $k$-stop protocol.}
    \label{fig:quantum_vs_classical}
\end{figure}

\section{Multiple sensors and resource hierarchy}\label{sec:multi}
Finally, we consider $M$ sensors. Each sensor carries an internal qubit and a spatial degree of freedom on the same $N = L/\Delta$ discrete sites. When the field occupies the (unknown) site $x_0$, the interaction Hamiltonian is
\begin{equation}
    H_{x_0}(t) = B\sum_{m=1}^{M}\sigma^{(m)}_{z} \otimes \ketbra{x_0}^{(m)} + H_c(t),\label{eq:Hmulti}
\end{equation}
where $H_c(t)$ denotes an arbitrary control Hamiltonian (which can implement coherent motion and gates) that does not depend on $x_0$. We ask for the minimum interrogation time $T$ required to reach a constant success probability, e.g., $P(T) \ge 2/3$, averaged over the unknown field location $x_0$. The achievable scaling depends on two independent resources: \emph{spatial coherence} (equivalently, superpositions over spatial modes/sites) and \emph{entanglement} between sensors during the protocol. The resulting ``resource hierarchy'' is summarized in~\cref{tab:scalings}.

\begin{table}[b]
    \centering
    \begin{tabular}{|c|c|c|}
        \hline
        Sensor &  Classical motion &  Quantum motion \\
        \hline 
        Entanglement  $\checkmark$  &  $\Theta({N}/{M})$ & $\Theta( {\sqrt{N}}/{M})$ \\
        \hline 
        Entanglement $\times$ & $\Theta({N}/{M })$ & $\Theta(\sqrt{{N}/{M}})$ \\
        \hline 
    \end{tabular}
    \caption{Asymptotic interrogation time $T$ to reach a constant average success probability, e.g., $P \ge 2/3$, when the field occupies one of $N$ sites and there are $M$ sensors.}
    \label{tab:scalings}
\end{table} 

\begin{theorem}[Multi-sensor optimal search time scaling laws]\label{thm:multi_scaling}
    For $M \le N$, the optimal interrogation time scalings in~\cref{tab:scalings} are achievable and optimal (up to constant factors) under the corresponding resource assumptions.
\end{theorem}

\noindent\emph{Proof sketch and protocols.}
We first describe protocols achieving these times and then discuss their optimality.

(i) \emph{No spatial coherence} (classical motion/placement): Without coherent superposition across sites, each sensor can only sample the unknown field location by occupying definite sites. A simple protocol places the $M$ sensors on $M$ distinct sites per round and repeats until successful detection; averaging over all possible $x_0$, this yields $P(T) \sim (MBT)/N$, hence $T = O(N/(MB))$ for this protocol.

(ii) \emph{Spatial coherence without entanglement}: Spatial superposition enables amplitude amplification within a chosen search subspace. With $M$ sensors that remain unentangled with each other, one may partition the $N$ sites into $M$ disjoint blocks of size approximately $N/M$, run the single-sensor coherent protocol of~\cref{sec:quantum} in parallel on each block, and finally measure; this achieves constant success in time
$T = O(\sqrt{N/M}/B)$.

(iii) \emph{Spatial coherence with entanglement}: One can prepare a GHZ-like internal state so that the oracle phase is amplified by a factor $M$, reducing the oracle time from $\pi/B$ to $\pi/(MB)$. Combining this $M$-fold faster oracle with Grover amplification over $N$ locations yields $T = O(\sqrt{N}/(MB))$.

Optimality of these protocols follows from bounding the success probability in~\cref{eq:P_succ}; see~\cref{app: grover_proof}~\footnote{Note that the proof in~\cref{app: grover_proof} also bounds the success probability for general state-discrimination, which upper bounds~\cref{eq:P_succ}.}. The rate of distinguishability growth is controlled by the standard deviation of the unknown local generator $H_{x_0}$ averaged over $x_0$. Without spatial coherence in the joint sensor-ancilla state, the averaged standard deviation is at most $\mathcal{O}(M/N)$, giving the $\Omega(N/(MB))$ lower bound; with spatial coherence but no entanglement between the sensors, the variance is at most $\mathcal{O}(M/N)$, giving $\Omega(\sqrt{N/M}/B)$; and with entanglement among the sensors, the variance can scale as $\mathcal{O}(M^2/N)$, giving $\Omega(\sqrt{N}/(MB))$. Full proofs are given in~\cref{app: grover_proof}.

\Cref{tab:scalings} highlights that spatial coherence is the resource responsible for the quadratic improvement from $N$ to $\sqrt{N}$, while entanglement between the sensor-ancilla systems can provide an additional $\sqrt{M}$ enhancement. However, without spatial coherence, the task remains a classical search over $N$ possible locations, and entanglement does not improve the scaling beyond the linear-in-$M$ speedup already obtained by probing $M$ locations in parallel. Finally, our lower bounds allow arbitrary ancilla systems, adaptive control, and general measurements, implying that the scalings in~\cref{tab:scalings} are optimal for these general strategies, provided that the joint sensor-ancilla states throughout the evolution satisfy the resource constraints of motion and entanglement.

\section{Discussion and outlook}\label{sec:outlook}
We introduced a field detection problem that isolates \emph{motional coherence} as a metrological resource. For a classically moving sensor we identified an optimal $k$-stop protocol and showed that allowing coherent spatial superposition qualitatively changes the picture: the unknown field region acts as an oracle, enabling a quadratic reduction in the interrogation time needed to reach high detection probability. For multiple sensors, the optimal scalings separate into two independent resources: spatial coherence reduces the effective search complexity ($N\to \sqrt{N}$), while entanglement can reduce the effective time per oracle call (from $1/\sqrt{M}B$ to $1/(MB)$). 

Although we have formulated our task in terms of field detection, the protocols considered here can also identify the field \emph{location} without changing their asymptotic interrogation-time scalings. For the classical $k$-stop protocol, this can be achieved by measuring the sensor after each stop, so that a positive detection identifies the corresponding field region. Similarly, the coherent-search protocols amplify the sensor's spatial population at the field location, allowing the location to be determined by a final position measurement.

Our results also extend directly to any search space that can be partitioned into $N$ candidate regions in which the field may reside. Since we do not impose locality constraints on the control Hamiltonian $H_c(t)$, only the number of candidate regions is relevant for the field detection task. Such higher-dimensional problems can therefore be represented as an effective one-dimensional search over $N$ regions.

Several extensions are immediate. First, realistic sensors experience dephasing and control errors, so it is important to quantify how robust the Grover-style enhancement is in the presence of noise, and to identify the crossover point where classical motion becomes preferable~\cite{leng2023quantum,zhou2018achieving,conlon2026fault, demkowicz2017adaptive}. Second, one may consider unknown field strength in addition to unknown location, or multiple disjoint field regions. Third, our model treats motion control abstractly via $H_c(t)$; it would be valuable to incorporate explicit dynamical constraints such as bounded velocity, finite acceleration, or locality of spatial control. Finally, coherent motion may be used to accelerate other spatially structured metrological tasks.

% There are several potential applications of moving sensors. For example, in the life sciences, detection of changes in local environments is often indicative of the presence of some pathology, such as disease, inflammation, or a tumor. There are many techniques that can be used to detect such pathologies, such as imaging techniques like magnetic resonance imaging (MRI) and positron emission tomography (PET) scans, but these expose patients to harmful radiation and so should be used sparingly. Instead, scanning a patient with a probe or sensor that can discern changes in the local physical environment without the need for exposure to radiation may be a viable alternative for pathology detection. For example, some tumors cause changes in temperature to the local environment~\cite{knapp2022tumor}, so sending a sensor moving through the bloodstream that is sensitive to temperature differentials could be used for cancer screening. Such an approach aligns well with the model we assume in this work: we suspect the presence of a signal in a localized region of space, but we are not sure exactly where, and we would just like to determine if there is such a signal or not. If the signal \emph{is} detected, then more advanced techniques can be employed to pinpoint exactly where.

\begin{acknowledgments}
The authors acknowledge the use of ChatGPT and Claude for improving proofs in Appendices B and C and for proofreading. The authors take full responsibility for the content of this publication. E.Z., E.A., A.F., D.J.S., and A.V.G.\ were supported in part by ARL (W911NF-24-2-0107), ONR MURI, AFOSR MURI, NSF STAQ program,  and NQVL:QSTD:Design:FTL. E.Z., E.A., A.F., D.J.S., A.M.C., and A.V.G.\ were supported in part by NSF QLCI (awards No.~OMA-2120757 and No.~OSI-2553574) and DoE ASCR Quantum Testbed Pathfinder program (award No.~DE-SC0024220). E.Z., E.A., A.F., D.J.S., A.M.C., and A.V.G. also acknowledge support from  the U.S.~Department of Energy, Office of Science, National Quantum Information Science Research Centers, Quantum Systems Accelerator (award No.~DE-SCL0000121) and the U.S.\ Department of Energy, Office of Science, Accelerated Research in Quantum Computing, Fundamental Algorithmic Research toward Quantum Utility (FAR-Qu). Y.-X.W.~acknowledges support from a QuICS Hartree Postdoctoral Fellowship. D.J.S.~acknowledges further funding support from a graduate research fellowship from the Joint Quantum Institute at the University of Maryland, College Park. S.R.M. is supported by the NSF QLCI award OMA-2120757.
\end{acknowledgments}

\appendix
\onecolumngrid
\clearpage

\section{Relation to state discrimination}\label{app:state_discrimination}
In the main text, we quantified the performance of a field-detection protocol by the probability of certifiably detecting the field. Here, we relate this figure of merit to standard notions in quantum state discrimination. Let the two hypotheses be
\begin{equation}
    H_0: \text{ no field},\qquad H_1: \text{ field present}.
\end{equation}
For a fixed field position $x_0$, the corresponding final states are
\begin{equation}
    \rho_0(T) = \ketbra{\psi_{\emptyset}(T)},\qquad \rho_1(T) = \ketbra{\psi_{x_0}(T)}.
\end{equation}
A binary measurement is described by a POVM $\{E_0, E_1\}$, where outcome $0$ means ``no field'' and outcome $1$ means ``field present.'' Thus, we have
\begin{equation}
    \Pr(1|0) = \Tr(E_1\rho_0),
    \qquad
    \Pr(0|1) = \Tr(E_0\rho_1).
\end{equation}
The goal is to minimize the error probability
\begin{equation}\label{eq:weighted_error}
    P_\mathrm{err} = \frac{1}{2}\big(\Pr(1|0) + \Pr(0|1)\big) .
\end{equation}
The minimum of~\cref{eq:weighted_error} is given by the Helstrom bound~\cite{Helstrom1976,Holevo1973,BarnettCroke2009},
\begin{equation}
    P_\mathrm{err}^\mathrm{Hel} = \frac{1}{2}\left(1-\frac{1}{2}\norm{\rho_0-\rho_1}_1\right),\label{eq:helstrom_error}
\end{equation}
and the corresponding success probability is
\begin{equation}
    P_\mathrm{succ}^\mathrm{Hel} = \frac{1}{2} + \frac{1}{4}\norm{\rho_0 - \rho_1}_1.\label{eq:helstrom_succ}
\end{equation}
For the pure states considered above, this becomes
\begin{equation}
    P_\mathrm{succ}^\mathrm{Hel}(x_0;T) = \frac{1}{2}\left( 1 + \sqrt{1 - |\braket{\psi_\emptyset(T)}{\psi_{x_0}(T)}|^2} \right).\label{eq:helstrom_pure}
\end{equation}
For field detection, however, one may want a stronger operational guarantee: whenever the measurement declares that a field is present, it should be correct. This imposes the ``zero-false-positive'' condition
\begin{equation}
    \Pr(1|0) = \Tr(E_1\rho_0) = 0.\label{eq:zero_false_alarm}
\end{equation}
More generally, one can allow a false alarm tolerance $\epsilon$ and maximize the detection probability
\begin{equation}
    P_{\rm D}^{\max}(\epsilon) = \max_{0\le E_1\le I}\Tr(E_1\rho_1)\quad \text{subject to} \quad \Tr(E_1\rho_0)\le \epsilon .\label{eq:neyman_pearson}
\end{equation}
This is known as the quantum Neyman-Pearson formulation of asymmetric hypothesis testing~\cite{Holevo1973,HayashiNagaoka2003}. The quantity used in the main text corresponds to the special case $\epsilon = 0$. If the no-field state is pure, then~\cref{eq:zero_false_alarm} implies that $E_1$ can only have support orthogonal to $\ket{\psi_{\emptyset}(T)}$. The optimal certified-detection measurement is therefore
\begin{equation}
    E_1^{\star} = I - \ketbra{\psi_{\emptyset}}.
\end{equation}
For a fixed $x_0$, this gives
\begin{equation}
    P_\mathrm{D}^{\max}(0) = \mel{\psi_{x_0}(T)}{E_1^\star}{\psi_{x_0}(T)} = 1 - \abs{\braket{\psi_\emptyset(T)}{\psi_{x_0}(T)}}^2.\label{eq:certified_detection}
\end{equation}
This is precisely the quantity denoted $p_{x_0}(T)$ in the main text. Finally, the Helstrom and certified-detection figures of merit are closely related. For pure states, they both increase as $|\!\braket{\psi_{\emptyset}(T)}{\psi_{x_0}(T)}$| decreases, and the two quantities are exactly related by
\begin{equation}
    P_\mathrm{succ}^\mathrm{Hel} = \frac{1}{2}\left( 1 + \sqrt{P_\mathrm{D}^{\max}(0)} \right),\label{eq:helstrom_certified_relation_pure}
\end{equation}
which ensures that $P_\mathrm{succ}^\mathrm{Hel} \geq P_\mathrm{D}^{\max}(0)$. 

\section{Proof of~\texorpdfstring{\cref{thm:classical_opt}}{Thm. 1} (optimal classical motion)}\label{app:optimal_classical}
In this appendix, we prove~\cref{thm:classical_opt}. We proceed in three steps:\\
\begin{enumerate}
    \item We solve the discretized problem on $N = L/\Delta$ equally spaced sites, where the field occupies exactly one site.
    \item We show the same optimal value holds on any refined lattice of $rN$ equally spaced sites where the field spans $r$ consecutive sites.
    \item We take $r \to \infty$, which yields the continuous formulation in~\cref{eq:Pclassical_general}.\\
\end{enumerate}
We begin with the discretized problem.

\subsection{Discrete \texorpdfstring{$N$}{N}-site problem (field occupies one site)}\label{app:optimal_classical_discreteN}
Consider $N$ sites labeled by index $j \in \{1, \dots, N\}$ and a field located at an unknown site $x_0 \in \{1, \dots, N\}$, chosen uniformly at random. A classical protocol (with negligible travel time) is specified by dwell times $\{\tau_j\}_{j=1}^{N}$ satisfying
\begin{equation}
    \tau_j \ge 0,\qquad \sum_{j=1}^{N}\tau_j = T.\label{eq:disc_constraints}
\end{equation}
If the field is at $x_0$, the sensor interacts for time $\tau_{x_0}$ and the success probability is $f(\tau_{x_0}) \coloneqq \sin^2(B\tau_{x_0})$. Averaging over $x_0$ gives
\begin{equation}
    P_\mathrm{cl}(T) = \frac{1}{N}\sum_{j=1}^{N} f(\tau_j) = \frac{1}{N}\sum_{j=1}^{N}\sin^2(B\tau_j).\label{eq:Pcl_discrete_app}
\end{equation}
Here, we assume that $T < \pi N/(2B)$. For $T \geq \pi N/(2B)$, we can achieve $P_\mathrm{cl}(T) = 1$ by choosing $\tau_j = \pi/(2B)$ for all $j$, and then measuring the sensor at time $\pi N/(2B)$.

\begin{lemma}[No wasted time beyond the first maximum]\label{lem:clip}
    Fix $T\le \frac{\pi N}{2B}$. There exists an optimal solution of~\cref{eq:disc_constraints,eq:Pcl_discrete_app} with
    \begin{equation}
        0 \le \tau_j \le \frac{\pi}{2B}\qquad \forall j \in \{1, \dots, N\}.\label{eq:clip}
    \end{equation}
\end{lemma}

\begin{proof}
    Suppose a feasible schedule has some $\tau_j > \pi/(2B)$. Find $\tau_j' \in  [0,\frac{\pi}{2B}]$ such that 
    \begin{equation}
        \sin^2(B\tau_j') = \sin^2(B\tau_j),\qquad \delta \coloneqq \tau_j-\tau_j' > 0.
    \end{equation}
    Thus, we have $f(\tau_j') = f(\tau_j)$, while the total used time decreases by $\delta$. Because $T \le \frac{\pi N}{2B}$, there must exist at least one index $i$ with $\tau_i < \pi/(2B)$ (otherwise $\sum_j \tau_j \ge N\pi/(2B)$). Allocate the freed time $\delta$ to such an index (or split it among several indices that remain below $\pi/(2B)$). Since $f'(u) = B\sin(2Bu) > 0$ for $u \in (0,\pi/(2B))$, this reallocation strictly increases the objective while maintaining the constraint~\cref{eq:disc_constraints}. Repeating this operation eliminates all $\tau_j > \pi/(2B)$.
\end{proof}

Next, we have the following:

\begin{lemma}[Two inequalities]\label{lem:two_ineq}
    The following inequalities hold for $u,v \in [0,\pi/2]$:
    \begin{alignat}{2}
        \sin^2 u + \sin^2 v &\le \sin^2(u+v),& u+v &\le \frac{\pi}{2},\label{eq:ineq_merge}\\
        \sin^2 u + \sin^2 v &\le 2\sin^2\left(\frac{u+v}{2}\right),& \quad u+v &\ge \frac{\pi}{2}.\label{eq:ineq_average}
    \end{alignat}
\end{lemma}

\begin{proof}
    We use the trigonometric identity
    \begin{equation}
        \sin^2 u+\sin^2 v = 1-\cos(u+v)\cos(u-v).
    \end{equation}
    First, suppose $u + v \le {\pi}/{2}$. Then, $\abs{u - v} \le u + v \le {\pi}/{2}$, so $\cos(u-v) \ge \cos(u+v)$. Since $\cos(u+v) \ge 0$, we get
    \begin{equation}
        \sin^2{u} + \sin^2{v} = 1 - \cos(u + v)\cos(u - v) \leq 1 - \cos^2(u+v) = \sin^2(u+v),
    \end{equation}
    which proves~\cref{eq:ineq_merge}. Next, suppose $u + v \ge {\pi}/{2}$. Then, $\cos(u+v) \le 0$, while $\cos(u-v) \le 1$. Therefore
    \begin{equation}
        \sin^2{u} + \sin^2{v} = 1 - \cos(u+v)\cos(u-v) \leq 1 - \cos(u+v) = 2\sin^2\left( \frac{u+v}{2} \right),
    \end{equation}
    which proves~\cref{eq:ineq_average} and proves the claim in the lemma.
\end{proof}

Finally, we have the following:

\begin{proposition}[Optimal value on $N$ sites]\label{prop:Nsites}
    For $0 \le T \le \frac{\pi N}{2B}$, the optimal value of~\cref{eq:Pcl_discrete_app} is
    \begin{equation}
        P_\mathrm{cl}^{\star}(T) = \max_{1 \le k \le N}\frac{k}{N}\sin^2\left(\frac{BT}{k}\right).
    \end{equation}
    For $T \ge \frac{\pi N}{2B}$, one can achieve $P_\mathrm{cl}^{\star}(T)=1$.
\end{proposition}

\begin{proof}
Let $\{\tau_j\}_{j=1}^N$ be an optimizer of~\cref{eq:Pcl_discrete_app} subject to~\cref{eq:disc_constraints}. By~\cref{lem:clip}, we assume that $0 \le \tau_j \le \frac{\pi}{2B}$ for all $j$. Next, suppose that two nonzero dwell times $\tau_i$ and $\tau_\ell$ satisfy $\tau_i + \tau_\ell \le {\pi}/{2B}$. Then, we can replace the pair $(\tau_i,\tau_\ell)$ by
\begin{equation}
    (\tau_i + \tau_\ell, 0).
\end{equation}
This preserves the total time $T$, and by~\cref{eq:ineq_merge} it does not decrease the objective. Using~\cref{lem:clip}, and repeating the same argument, we conclude that there exists an optimizer $\{\tau_j\}_{j=1}^N$ such that every two distinct nonzero dwell times satisfy
\begin{equation}
    \tau_i+\tau_\ell > \frac{\pi}{2B}.
\end{equation}
Now suppose that two nonzero dwell times $\tau_i$ and $\tau_\ell$ are not equal. Since their sum is larger than $\frac{\pi}{2B}$,~\cref{eq:ineq_average} implies that replacing them by
\begin{equation}
    \left( \frac{\tau_i + \tau_\ell}{2}, \frac{\tau_i + \tau_\ell}{2} \right)
\end{equation}
does not decrease the objective. In fact, because $\tau_i \neq \tau_\ell$, the inequality is strict, so this replacement strictly increases the objective. This contradicts optimality. Hence, all nonzero dwell times must be equal. Therefore, there exists an optimal set of times with exactly $k$ nonzero elements, all equal to $T/k$, for some $1 \le k\le N$. Thus,
\begin{equation}
    P_\mathrm{cl}^{\star}(T) = \max_{1 \le k \le N} \frac{k}{N}\sin^2\left(\frac{BT}{k}\right),
\end{equation}
which completes the proof.
\end{proof}

\subsection{Refined lattice: \texorpdfstring{$rN$}{rN} sites with a field spanning \texorpdfstring{$r$}{r} consecutive sites}\label{app:optimal_classical_refine}
Now, we discretize the ring into $rN$ equally spaced sites labeled $i \in\{0, 1, \dots, rN-1\}$, where $r$ is a positive integer. The field occupies $r$ consecutive sites starting at an unknown $i_0$, and we assume no prior information is available about $i_0$ so that it satisfies a probability distribution uniform among all $rN$ sites. A classical schedule specifies dwell times $\{t_i\}_{i=0}^{rN-1}$ with $t_i \ge 0$ and $\sum_i t_i = T$. If the field starts at $i_0$, the total interaction time is
\begin{equation}
    s_{i_0} \coloneqq \sum_{j=0}^{r-1} t_{i_0+j},
\end{equation}
where the indices are taken modulo $rN$. Hence the averaged success probability is
\begin{equation}
    P_\mathrm{cl}^{(r)}(T) = \frac{1}{rN}\sum_{i_0=0}^{rN-1}\sin^2\left(B s_{i_0}\right).\label{eq:Pcl_rN}
\end{equation}
With this, we have the following:

\begin{proposition}
\label{prop:refinement_invariance}
    For every integer $r \ge 1$, the optimal value of~\cref{eq:Pcl_rN} equals the $N$-site optimum:
    \begin{equation}
        \max_{\{t_i\}} P_\mathrm{cl}^{(r)}(T) = \max_{1\le k\le N}\frac{k}{N}\sin^2\!\left(\frac{BT}{k}\right),\qquad 0\le T\le \frac{\pi N}{2B}.\label{eq:refinement_value}
    \end{equation}
\end{proposition}

\begin{proof}
    Group the starting positions $i_0$ according to their residue class modulo $r$:
    \begin{equation}
        P_\mathrm{cl}^{(r)}(T) = \frac{1}{r}\sum_{l=0}^{r-1}\left[\frac{1}{N}\sum_{m=0}^{N-1}\sin^2\left(Bs_{l+mr}\right)\right].
    \end{equation}
    For a fixed $l$, we can write
    \begin{equation}
        \sum_{m=0}^{N-1}s_{l+mr} = \sum_{i=0}^{rN-1} t_i = T.
    \end{equation}
    Thus, the variables $\{s_{l+mr}\}_{m=0}^{N-1}$ form a feasible $N$-site schedule with total time $T$. By~\cref{prop:Nsites},
    \begin{equation}
        \frac{1}{N}\sum_{m=0}^{N-1}\sin^2\left(Bs_{l+mr}\right)\le \max_{1 \le k \le N}\frac{k}{N}\sin^2\left(\frac{BT}{k}\right).
    \end{equation}
    Next, averaging over $l$ gives
    \begin{equation}
        P_\mathrm{cl}^{(r)}(T) \le \max_{1 \le k \le N}\frac{k}{N}\sin^2\left(\frac{BT}{k}\right).
    \end{equation}
    For the matching lower bound, fix $k \in\{1, \dots, N\}$ and choose
    \begin{equation}
        t_{qr-1} = \frac{T}{k},\qquad q = 1, \dots, k,
    \end{equation}
    with all other $t_i = 0$. For each residue class $l$, the site $qr-1$ belongs to the block starting at $l+(q-1)r$. Hence, for every $l$, exactly $k$ of the block-sums $\{s_{l+mr}\}_{m=0}^{N-1}$ are equal to $T/k$, and the rest are zero. Therefore,
    \begin{equation}
        P_\mathrm{cl}^{(r)}(T) = \frac{k}{N}\sin^2\left(\frac{BT}{k}\right).
    \end{equation}
    Maximizing over $k$ proves the claim.
\end{proof}

\subsection{Continuum limit \texorpdfstring{($r\to\infty$)}{}}\label{app:optimal_classical_continuum}
We now connect the refined discrete model to the continuous occupation-time density $\tau(x)$ used in~\cref{eq:Pclassical_general}. For a refinement factor $r$, partition the ring into $rN$ intervals (cells) of length $\delta = \Delta/r$. Given a density $\tau(x)$, define the cell times as
\begin{equation}
    t_i^{(r)} \coloneqq \int_{i\delta}^{(i+1)\delta}\tau(x)\mathop{dx},\qquad i = 0, \dots, rN-1.\label{eq:celltimes}
\end{equation}
Then, $t_i^{(r)} \ge 0$ and $\sum_i t_i^{(r)} = \int_0^L\tau(x)\mathop{dx} = T$, so $\{t_i^{(r)}\}$ is a feasible schedule for the refined lattice. Moreover, the refined-lattice interaction time when the field
starts at cell $i_0$ is
\begin{equation}
    s_{i_0}^{(r)} = \sum_{j=0}^{r-1} t_{i_0+j}^{(r)} = \int_{i_0\delta}^{i_0\delta+\Delta}\tau(x)\mathop{dx},
\end{equation}
and the refined-lattice objective~\cref{eq:Pcl_rN} is the approximation of the continuous average in~\cref{eq:Pclassical_general}:
\begin{equation}
    P_\mathrm{cl}^{(r)}(T) = \frac{1}{rN}\sum_{i_0=0}^{rN-1}\sin^2\left(B s_{i_0}^{(r)}\right) \xrightarrow{r\to\infty} \frac{1}{L}\int_0^L \sin^2\left(B\int_{x_0}^{x_0+\Delta}\tau(x)\mathop{dx}\right)\mathop{dx_0}.\label{eq:Riemann_conv}
\end{equation}
Since~\cref{prop:refinement_invariance} shows that the optimum of $P_\mathrm{cl}^{(r)}(T)$ is the same for all $r$, taking $r \to \infty$ implies the continuous optimum equals the discrete optimum, establishing~\cref{thm:classical_opt} for the model of~\cref{eq:Pclassical_general}.

\subsection{Large-\texorpdfstring{$N$}{N} limit}\label{app:optimal_classical_linear}
For large $N$, it is useful to treat $k$ as approximately continuous and maximize
\begin{equation}
    P_k(T) = \frac{k}{N}\sin^2\left(\frac{BT}{k}\right).
\end{equation}
Setting $x \coloneqq BT/k$ (so $k = BT/x$) gives
\begin{equation}
    P_k(T) = \frac{BT}{N}\frac{\sin^2 x}{x}.
\end{equation}
The maximum of this expression occurs at $x^* \simeq 
1.1656$. In the intermediate-time window ($1 \lesssim k^\star \lesssim N$, i.e., $x_0/B \lesssim T \lesssim x^* N/B$), we obtain
\begin{equation}
    k^\star \simeq \frac{BT}{x^*},\qquad P_\mathrm{cl}^\star(T) \simeq \frac{\sin^2 x^*}{x^*}\frac{BT}{N} \approx 0.7246\frac{BT}{N},
\end{equation}
where $P_\mathrm{cl}^\star(T)$ scales linearly with time $T$. This result can also be interpreted in terms of an effective sensor velocity. The sensor spends time ${T}/{k^\star}$ at each visited site before moving a distance $\Delta$ to the next one. Hence, the corresponding optimal velocity is
\begin{equation}
    v^\star \simeq \frac{\Delta}{T/k^\star} \approx 0.8580B\Delta.
\end{equation}
Therefore, in the large-$N$ limit and within the intermediate-time regime, the
optimal strategy is well-approximated by moving the sensor at a constant
velocity.

\section{Constant-velocity protocol}\label{app:const_velocity}
In this appendix, we analyze the protocol in which the sensor moves at a \emph{constant} velocity $v \ge 0$ around the ring for a total time $T$, after which we measure the internal qubit in the $\{\ket{\pm}\}$ basis. We derive an analytic expression for the averaged success probability $P_\mathrm{cv}(v,T)$ and explain how~\cref{fig:cv_opt} (the optimal probability and optimal velocity versus time) is obtained by maximizing over $v$. We also show that while constant-velocity scanning is generally suboptimal at finite $N$, its optimal performance converges to the $k$-stop protocol as we take $N \to\infty$.

\subsection{Average success probability}\label{app-subsec:avg-succ-prob}
Fix the field segment to be $[0,\Delta]$ and average over the sensor's initial position $s \in [0,L)$ uniformly; this is equivalent to fixing the sensor start and averaging over the unknown field position $x_0$. Under constant velocity $v$,
\begin{equation}
    x(t) = (s+vt)\pmod{L},\qquad t \in[0,T].
\end{equation}
Define the total distance traversed $a \coloneqq vT$ so that the sensor completes $q \coloneqq \left\lfloor\frac{a}{L}\right\rfloor$ full laps and then travels a remaining arc of length $\ell \coloneqq a -qL \in [0,L)$. Each full lap contributes a deterministic interaction time $\Delta/v$, hence
\begin{equation}
    \tau(s;v,T) = \frac{q\Delta}{v} + \frac{1}{v}\Lambda(s;\ell),\label{eq:tau_decomp}
\end{equation}
where $\Lambda(s;\ell) \in [0,\Delta]$ is the \emph{spatial overlap length} between the field interval $[0,\Delta]$ and the arc segment of length $\ell$ traversed in the final partial lap starting at $s$. Equivalently, $\Lambda(s;\ell)$ is the overlap length of two arcs of lengths $\Delta$ and $\ell$ on a circle of circumference $L$, averaged over their relative shift. The success probability conditioned on $s$ is therefore
\begin{equation}
    p(s;v,T) = \sin^2\left(B\,\tau(s;v,T)\right) = \sin^2\left(\frac{B}{v}\left(q\Delta+\Lambda(s;\ell)\right)\right),
\end{equation}
and the averaged success probability is
\begin{equation}
    P_\mathrm{cv}(v,T) = \frac{1}{L}\int_{0}^{L} p(s;v,T)\mathop{ds}.\label{eq:Pcv_def}
\end{equation}

\subsection{Closed form expression for \texorpdfstring{$P_\mathrm{cv}(v,T)$}{P(v,T)}}
It is convenient to write $\sin^2{x} =\frac{1}{2}\left(1-\cos 2x\right)$. Let
\begin{equation}
    \gamma \coloneqq \frac{2B}{v},\qquad \phi \coloneqq \gamma q\Delta = \frac{2Bq\Delta}{v}.\label{eq:gamma_phi}
\end{equation}
Then, \cref{eq:Pcv_def} becomes
\begin{equation}
    P_\mathrm{cv}(v,T) = \frac{1}{2}-\frac{1}{2L}\int_{0}^{L}\cos\left(\phi+\gamma\,\Lambda(s;\ell)\right)\mathop{ds} = \frac{1}{2} - \frac{1}{2L}\left[\cos\phi C(\ell) - \sin\phi S(\ell)\right],\label{eq:Pcv_cos_sin}
\end{equation}
where we defined the overlap integrals
\begin{equation}
    C(\ell) \coloneqq \int_{0}^{L}{\cos(\gamma\Lambda(s;\ell))\mathop{ds}},\qquad S(\ell) \coloneqq \int_{0}^{L}{\sin(\gamma\Lambda(s;\ell))\mathop{ds}}.\label{eq:CS_def}
\end{equation}
We now evaluate $C(\ell)$ and $S(\ell)$.

\paragraph*{Case I: $\ell\le L-\Delta$ (non-forced overlap).} Define
\begin{equation}
    m \coloneqq \min\{\ell,\Delta\},\qquad d \coloneqq \abs{\ell-\Delta}.
\end{equation}
As $s$ varies uniformly, the overlap length $\Lambda(s;\ell)$ (i) equals $0$ over a total shift-length $L - \ell - \Delta$, (ii) increases linearly from $0$ to $m$ over a shift-length $m$, (iii) stays at $m$ over a shift-length $d$, and (iv) decreases linearly back to $0$ over another shift-length $m$. Thus,
\begin{equation}
    C(\ell) = (L - \ell - \Delta) \cdot 1 + d\cos(\gamma m) + 2\defint{0}{m}{\cos(\gamma x)\mathop{dx}} = (L - \ell - \Delta) + d\cos(\gamma m) + \frac{2}{\gamma}\sin(\gamma m),\label{eq:C_case1}
\end{equation}
and
\begin{equation}
    S(\ell) = d\sin(\gamma m) + 2\defint{0}{m}{\sin(\gamma x)\mathop{dx}} = d\sin(\gamma m) + \frac{2}{\gamma}(1 - \cos(\gamma m)).\label{eq:S_case1}
\end{equation}
Substituting~\cref{eq:C_case1,eq:S_case1} into~\cref{eq:Pcv_cos_sin} then gives a closed form for $P_\mathrm{cv}(v,T)$ whenever $\ell \le L - \Delta$.

\paragraph*{Case II: $\ell > L-\Delta$ (forced overlap).}
In this regime, the traversed arc and the field arc must overlap for every shift $s$. Define the minimal overlap length
\begin{equation}
    c \coloneqq \ell + \Delta - L > 0,
\end{equation}
and the complementary arc lengths
\begin{equation}
    \ell' \coloneqq L - \ell,\qquad \Delta' \coloneqq L - \Delta.
\end{equation}
Let $G(s)$ be the overlap length of two arcs of lengths $\ell'$ and $\Delta'$; since
$\ell' + \Delta' = 2L - (\ell + \Delta) < L$, $G(s)$ has the same triangle--plateau--triangle structure as in Case I. We verify the identity
\begin{equation}
    \Lambda(s;\ell) = c + G(s),
\end{equation}
and therefore
\begin{align}
    C(\ell) &= \defint{0}{L}{\cos(\gamma(c + G(s)))\mathop{ds}} = \cos(\gamma c)C_G - \sin(\gamma c)S_G,\label{eq:C_case2}\\
    S(\ell) &= \defint{0}{L}{\sin(\gamma(c + G(s)))\mathop{ds}} = \sin(\gamma c)C_G + \cos(\gamma c)S_G,\label{eq:S_case2}
\end{align}
where
\begin{align}
    C_G &\coloneqq \defint{0}{L}{\cos(\gamma G(s))\mathop{ds}} = (L - \ell' - \Delta') + d'\cos(\gamma m') + \frac{2}{\gamma}\sin(\gamma m'),\\
    S_G &\coloneqq \defint{0}{L}{\sin(\gamma G(s))\mathop{ds}} = d'\sin(\gamma m') + \frac{2}{\gamma}(1 - \cos(\gamma m')),
\end{align}
with $m' \coloneqq \min\{\ell', \Delta'\}$ and $d' \coloneqq \abs{\ell'-\Delta'}$. Substituting~\cref{eq:C_case2,eq:S_case2} into~\cref{eq:Pcv_cos_sin} yields $P_\mathrm{cv}(v,T)$ for $\ell > L-\Delta$.

\subsection{Optimizing over \texorpdfstring{$v$}{v}}\label{app:cv_optimize}
For each interrogation time $T$, we define
\begin{equation}
    P_\mathrm{cv}^\star(T) \coloneqq \max_{v\ge 0} P_\mathrm{cv}(v,T),\qquad v^\star(T) \in \arg\max_{v \ge 0}P_\mathrm{cv}(v,T),\label{eq:optdefs}
\end{equation}
where $P_\mathrm{cv}(v,T)$ is the constant-velocity-averaged success probability in~\cref{eq:Pcv_def} (equivalently, the closed-form expressions obtained from~\cref{eq:Pcv_cos_sin,eq:C_case1,eq:S_case2}). In~\cref{fig:cv_opt}, we plot (a) $P_\mathrm{cv}^\star(T)$ and (b) the corresponding optimizer $v^\star(T)$ as a function of time $T$. The maximized curves in~\cref{fig:cv_opt} show three different regimes that mirror the $k$-stop protocol of~\cref{thm:classical_opt}. Note that constant velocity over time $T$ sweeps a total distance $vT$, which corresponds to scanning $k \approx {vT}/{\Delta}$ field-sized cells. Thus, the stationary strategy ($v=0$) corresponds to $k=1$, while scanning one full loop corresponds to $k=N$.

\begin{enumerate}
    \item \emph{Stationary regime (short times).} For small $T$, the optimal constant-velocity strategy is to remain stationary,
    \begin{equation}
        v^\star(T) = 0,\qquad P_\mathrm{cv}^\star(T) = \frac{1}{N}\sin^2(BT),
    \end{equation}
    which is the $k=1$ of the $k$-stop protocol. Intuitively, when $BT \lesssim \pi/2$, it is best to concentrate the interaction time rather than spread it over multiple locations.
    
    \item \emph{Intermediate regime (linear $P^\star$ and constant $v^\star$).} As $T$ increases, the optimal strategy begins to scan multiple candidate locations. As shown in~\cref{app:optimal_classical_linear}, in the large-$N$ limit, we can treat $k/N$ as a continuous variable. Maximizing $\frac{k}{N}\sin^2(BT/k)$ yields
    \begin{equation}
        v^\star(T) \simeq 0.8580B\Delta,\qquad P_\mathrm{cv}^\star(T) \simeq 0.7246\frac{BT}{N},\label{eq:cv_linear_plateau}
    \end{equation}
    that is, a linear-in-$T$ growth of the optimal success probability and an approximately constant optimal velocity.
    
    \item \emph{One-loop regime (long times).} For long times, the optimum occurs at $k= N$, corresponding to scanning the entire ring. Any velocity satisfying $vT \in L\mathbb{Z}$ yields a deterministic interaction time $T/N$ independent of the unknown field position, and therefore
    \begin{equation}
        v^\star(T) \approx \frac{L}{T},\qquad P_\mathrm{cv}^\star(T) = \sin^2\left(\frac{BT}{N}\right),
    \end{equation}
    where we choose the minimal velocity (which traverses one full lap) such that $vT \in L\mathbb{Z}$ as an optimal velocity. This explains the $v^\star(T) \sim 1/T$ decay in~\cref{fig:cv_opt}(b) and the approach of~\cref{fig:cv_opt}(a) to the uniform-coverage curve.
\end{enumerate}

\begin{figure}[t]
    \centering
    \subfloat[]{\includegraphics[width=0.48\linewidth]{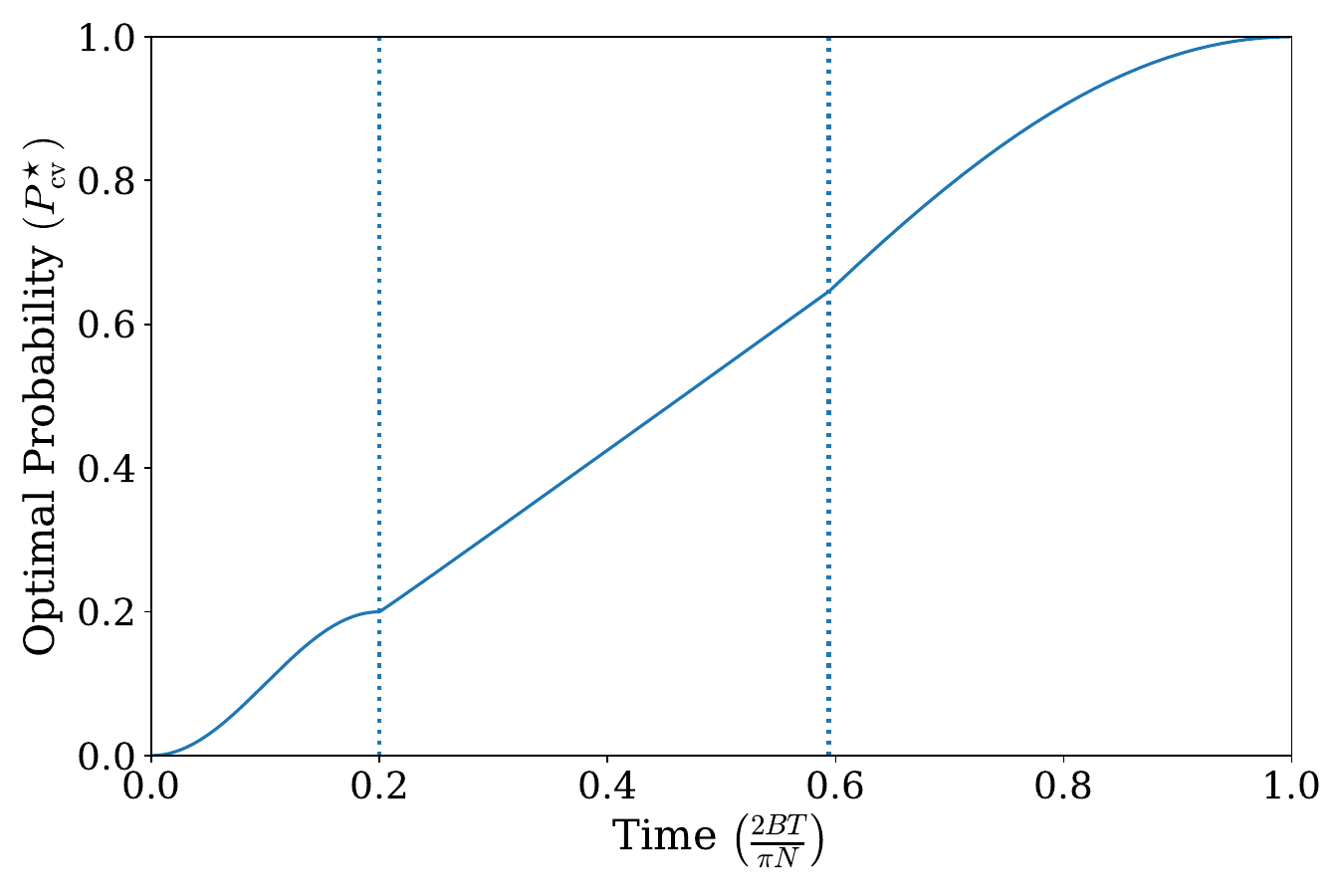}\label{fig:cv_opt_prob}}
    \hfill
    \subfloat[]{\includegraphics[width=0.48\linewidth]{v_const_v.pdf}\label{fig:cv_opt_vel}}
    \caption{Constant-velocity protocol for $N=5$. (a) Optimal success probability $P_\mathrm{cv}^\star(T) = \max_{v} P_\mathrm{cv}(v,T)$ versus interrogation time $T$. (b) Corresponding optimal velocity $v^\star(T)$. The curve exhibits three regimes: a short-time stationary regime ($v^\star = 0$), an intermediate regime with approximately constant optimal speed, and a long-time uniform-coverage regime with $v^\star(T) \approx L/T$.}
    \label{fig:cv_opt}
\end{figure}

\section{Proof of the lower bound for the moving sensor problem }\label{app: grover_proof}
In this final appendix, we provide proofs for the lower bounds in~\cref{tab:scalings} for the moving sensor problem when including sensors with quantum motion and/or entanglement. We first note that, classically, using $M$ sensors is approximately equivalent to extending the total sensing time to $MT$ in the regime $BT \gg 1$. In the quantum regime, however, the sensors may be in a quantum superposition of different spatial basis states, and multiple sensors can be in an entangled state, which may change the optimal $T$ scaling. We discuss those distinct cases in turn below.

Our proofs use techniques closely related to the hybrid method used to prove lower bounds for discrete-~\cite{bennett1997strengths} and continuous-time \cite{farhi1998analog} unstructured search. Those results lower bound the query complexity by bounding the Euclidean distance, which upper-bounds the trace distance and hence the Helstrom bound on single-shot distinguishability, between the evolution under any given oracle and a reference evolution without an oracle, averaged over all possible oracles. 
Our proofs differ in a crucial technical aspect: in contrast to unstructured search, where the oracle acts on a single query register, here the field acts on all $M$ sensors simultaneously, and the sensors' basis states can naturally be labeled by states of $N$ bosonic modes with at most $M$ excitations, where each sensor in the internal state $\ket{1}$ contributes one excitation on its site. As a result, the oracle couples to all sensors at the oracle site rather than to a single query register, which modifies the key step of the hybrid argument quantitatively and enables an additional entanglement advantage when multiple sensors can occupy the same site.

In what follows, we consider arbitrary $M$-sensor protocols. However, all our proofs can be directly generalized to strategies involving ancillas and arbitrary sensor-ancilla controls, because the key steps in the proof are only concerned with the action of the oracle Hamiltonian on the sensor(-ancilla) state and do not depend on whether the sensors are entangled with other ancillas or not. For protocols without spatial coherence and/or entanglement, these constraints are imposed on the joint state of the sensors and their ancillas [see discussion below \cref{appeq:T.lb.prod} for detail].

To simplify notation, we use the following labeling scheme for the state space of the sensors. First, we introduce a fictitious bosonic mode on each site. Then, we map an $M$-qubit sensor basis state with $n_{x}$ excitations on the $x$th site to the bosonic Fock state $\ket{\{n_{x}\}}$. $\mathcal{H}_{M}$ describes the subspace spanned by states with at most $M$ excitations. If we allow multiple sensors to stack on the same site, then the subspace spanned by all the possible sensor states is $\mathcal{H}_{M}$. If there is a field, then we denote the location of that field as site $x_{0}$ with $x_{0} \in \{1, 2, \ldots, N\} $. The corresponding Hamiltonian in the original basis reads $H_{x_{0}} = B\sum_{m=1}^{M}\sigma^{(m)}_{z} \otimes \ketbra{x_0}^{(m)}$, which maps to $H ^{(\text{eff})}_{x_{0}} = 2B a_{x_{0}}^{\dag}a_{x_{0}}$ on the sensor mode on the $x_{0}$th site up to a spin-independent spatial potential. Note that all the bounds below also hold for the full Hamiltonian, since in the original basis its magnitude is likewise bounded by $B$ times the number of sensors on site $x_0$. As such, we focus on the given bosonic Hamiltonian to simplify notation. We do not know where $x_{0}$ is before we start the protocol. 
We divide the total evolution into $K$ equal partitions, with total evolution time $T = K \Delta t$ and we take the $K \to \infty$, $\Delta t \to 0$ limit. We can generally describe a sensing strategy (conditioned on the field being located at an unknown site $x_{0}$) as 
\begin{equation}
    \ket{\Psi_{x_{0}}} = V_{K+1} \mathcal{T} \prod_{j=1} ^{K} e ^{-i H ^{(\text{eff})}_{x_{0}} \Delta t } V_{j} \ket{\psi_{0}},\qquad \ket{\Psi_{0}} = V_{K+1} \mathcal{T} \prod_{j=1} ^{K}  V_{j} \ket{\psi_{0}},
\end{equation}
where
\begin{equation}
    \mathcal{T} \prod_{j=1} ^{K}  V_{j} = \hat{V}_K\cdots\hat{V}_1.
\end{equation}
Thus, a generic strategy can be related to solving the following decision problem:
\begin{align}
    \text{distinguish between } \ket{\Psi_{x}}\text{ and }\ket{\Psi_{0}}.
\end{align}
The maximal probability of successfully distinguishing between two states is given by the Holevo-Helstrom bound (see~\cref{eq:helstrom_pure} in~\cref{app:state_discrimination}):
\begin{equation}
    p_{x_0} = \frac{1 + \frac{1}{2}\norm{\ketbra{\Psi_{x_0}} - \ketbra{\Psi_0}}_1}{2}.
\end{equation}
Suppose we have the optimal strategy, such that $p_{x_{0}} > 2/3$ for all $x_{0}$. Then, we have the following necessary condition:
\begin{align}
    \frac{1}{N}\defsum{x_0=1}{N}{p_{x_0}} = \frac{1}{N}\defsum{x_0=1}{N}{\frac{1 + \frac{1}{2}\norm{\ketbra{\Psi_{x_0}} - \ketbra{\Psi_0}}_1}{2}} &> \frac{2}{3}\\
    \implies \frac{1}{N}\defsum{x_0=1}{N}{\norm{\ketbra{\Psi_{x_0}} - \ketbra{\Psi_0}}_1} &> \frac{2}{3}.
\end{align}
Note that for a continuous target coordinate $x_{0} \in [0,L)$, we can similarly rewrite the above condition as
\begin{align}
\int ^{L}_{0} \frac{d x_{0} }{L} 
{\norm{\ketbra{\Psi_{ {x_{0}}}} - \ketbra{\Psi_0}}_1} &> \frac{2}{3}
.
\end{align}
In the continuous case, the projector $\ketbra{x_0}$ in $H_{x_0}$ is replaced by $\Pi_{x_0}=\int_{x_0}^{x_0+\Delta}\ketbra{x}dx$, which satisfies $\int_0^L\Pi_{x_0}\,dx_0=\Delta\,\mathbb{I}$, the continuous analog of $\sum_{x_0}\ketbra{x_0}=\mathbb{I}$. All the derivations below can therefore be directly generalized to the continuous regime (with ``no stacking'' understood as a minimum separation $\Delta$ between sensors).
Thus, to simplify notation, we focus on the discrete case in what follows, with the understanding that $1/N$ should be replaced by $\Delta / L $ in the final lower bounds when $x$ is continuous.

From the triangle inequality, we have
\begin{align}
    \frac{1}{N}\defsum{x_0=1}{N}{\norm{\ketbra{\Psi_{x_0}} - \ketbra{\Psi_0}}_1} &\leq \frac{2}{N}\defsum{x_0=1}{N}{\norm{\mathcal{T}\prod_{j=1}^{K}{e^{-iH_{x_0}^{(\mathrm{eff})}\Delta t}V_j\ket{\psi_0}} - \mathcal{T}\prod_{j=1}^{K}{e^{-i\langle H_{x_0}^{(\mathrm{eff})}(j\Delta t)\rangle\Delta t}V_j\ket{\psi_0}}}_2}\\
    &\leq \frac{2}{N}\defsum{x_0=1}{N}{\defsum{\ell=0}{K-1}{\norm{(e^{-iH_{x_0}^{(\mathrm{eff})}\Delta t} - e^{-i\langle H_{x_0}^{(\mathrm{eff})}(\ell\Delta t)\rangle\Delta t})\mathcal{T}\prod_{j=1}^{\ell}{V_j\ket{\psi_0}}}_2}}\label{neq:Hbound.time.gen}\\
    &\leq \frac{2K}{N}\max_{\ket{\Tilde{\psi}} \in \mathcal{H}_M}{\defsum{x_0=1}{N}{\norm{(e^{-i(H_{x_0}^{(\mathrm{eff})}-\langle\Tilde{\psi}| H_{x_0}^{(\mathrm{eff})}|\Tilde{\psi}\rangle)\Delta t} - \mathbb{I})\ket{\Tilde{\psi}}}_2}},
\end{align}
where we denote $\langle H_{x_0}^{(\mathrm{eff})}(\ell\Delta t)\rangle$ as the expectation value of $H_{x_0}^{(\mathrm{eff})}$ in the instantaneous reference state $\mathcal{T}\prod_{j=1}^{\ell}{V_j\ket{\psi_0}}$.
In the $\Delta t \to 0$ limit, $B \Delta t \ll 1 $, so we can write the above expression as
\begin{align}
&{}\frac{2 K }{N} \max_{\ket{\tilde{\psi} } \in \mathcal{H}_{M} }
\sum_{x_{0} =1} ^{N}  
\left \| ( e ^{-i (H ^{(\text{eff})}_{x_{0}}-\langle\Tilde{\psi}| H_{x_0}^{(\mathrm{eff})}|\Tilde{\psi}\rangle) \Delta t } - \mathbb{I} )  \ket{ \tilde{\psi}} \right  \|_{2}
\nonumber\\
&= \frac{2 \sqrt{2}K }{N} \max_{\ket{\tilde{\psi} } \in \mathcal{H}_{M} }
\sum_{x_{0} =1} ^{N}  
\sqrt{1- \left | 
% \braket{\tilde{\psi} |  e ^{-i H ^{(\text{eff})}_{x_{0}} \Delta t }  | \tilde{\psi}  } 
\mel{\tilde{\psi} }{e ^{-i (H ^{(\text{eff})}_{x_{0}}-\langle\Tilde{\psi}| H_{x_0}^{(\mathrm{eff})}|\Tilde{\psi}\rangle) \Delta t } }{\tilde{\psi} }
\right |
} 
\\
&\le \frac{2 \sqrt{2}K }{N} \max_{\ket{\tilde{\psi} } \in \mathcal{H}_{M} }
\sqrt{N \sum_{x_{0} =1} ^{N} \left ( 1-  \left | 
% \braket{\tilde{\psi} |  e ^{-i H ^{(\text{eff})}_{x_{0}} \Delta t }  | \tilde{\psi}  } 
\mel{\tilde{\psi}}{e ^{-i (H ^{(\text{eff})}_{x_{0}}-\langle\Tilde{\psi}| H_{x_0}^{(\mathrm{eff})}|\Tilde{\psi}\rangle) \Delta t } }{\tilde{\psi}}
\right | \right ) }
\\
&= 2 \sqrt{2} K 
\sqrt{\frac{\max_{\ket{\tilde{\psi} } \in \mathcal{H}_{M} } \sum_{x_{0} =1} ^{N} \left ( 1- \left | 
% \braket{\tilde{\psi} |  e ^{-i H ^{(\text{eff})}_{x_{0}} \Delta t }  | \tilde{\psi}  } 
\mel{\tilde{\psi} }{e ^{-i (H ^{(\text{eff})}_{x_{0}}-\langle\Tilde{\psi}| H_{x_0}^{(\mathrm{eff})}|\Tilde{\psi}\rangle) \Delta t } }{\tilde{\psi}}
\right | \right )
}{N}   }
\\
&\le 2 K \Delta t 
\sqrt{\frac{\max_{\ket{\tilde{\psi} } \in \mathcal{H}_{M} } \sum_{x_{0} =1} ^{N} \left [ 
% \braket{\tilde{\psi} |  (H ^{(\text{eff})}_{x_{0}} ) ^{2}  | \tilde{\psi}  }
\mel{\tilde{\psi}}{(H ^{(\text{eff})}_{x_{0}} ) ^{2} }{\tilde{\psi} }
- 
% \braket{\tilde{\psi} | H ^{(\text{eff})}_{x_{0}}  | \tilde{\psi}  }
\mel{\tilde{\psi} }{H ^{(\text{eff})}_{x_{0}} }{\tilde{\psi}}^{2}  \right ] 
}{N}   }
\\
&\le 4 BT \frac{M  }
{\sqrt{N}}
. 
\end{align}
If we are allowed to place multiple sensors at a single site, then the optimal state is a $W$ state made of an effective qudit. Each $M$-dimensional qudit is formed by putting $M$ qubit sensors on the same site and taking the permutationally invariant basis. This gives  
\begin{align}
    \ket{\tilde{\psi}} \in \mathcal{H}_{M}^{\mathrm{stack}}:\quad T = \Theta\left(\frac{\sqrt{N}}{M}\right).
\end{align}
If we do not allow stacking multiple sensors on a single site, then we can use a similar argument to show that  
\begin{align}
    \ket{\tilde{\psi}} \in \mathcal{H}_{M}^{\mathrm{dist}}:\quad T = \Omega\left(\sqrt{\frac{N}{M}}\right).
\end{align}
An $M$-qubit GHZ state concatenated with an $(N/M)$-branch $W$ state achieves this bound. The same bound holds for unentangled sensor-ancilla states even if the sensors may share a site, since the variance is then additive over sensors, giving $\sum_{x_0}\mathrm{Var}(H_{x_0})\le 4B^2\sum_{x_0}\langle n_{x_0}\rangle\le 4B^2M$. We thus obtain the optimal scaling for quantum sensors without stacking as
\begin{align}
    \ket{\tilde{\psi}} \in \mathcal{H}_{M} ^{\text{dist}}: \quad T = \Theta(\sqrt{N/M}).
\end{align}
Here, we assume that $N > M$. In the regime $M > N$, adding more sensors does not further reduce the sensing time, because we have assumed that there is at most a single qubit on each site.

We can use a similar argument as above to derive a query complexity lower bound for strategies without spatial coherence. By this, we mean protocols where each sensor occupies a definite site \emph{throughout} the sensor's evolution. In contrast, if we begin with a fixed-position initial state but allow coherent superposition between different position-basis states during the sensing protocol, then we would be able to saturate the general lower bound in protocol time $T$ discussed above.

As the trace norm is convex, it suffices to prove a lower bound for all pure probe states. We focus on strategies without spatial coherence, where each sensor has a definite position during the entire protocol, while the internal and ancilla state can be arbitrary.
For any given sensor initial state and control protocol, the upper bound on success probability in~\cref{neq:Hbound.time.gen} still applies:
\begin{equation}
    \frac{1}{N}\defsum{x_0=1}{N}{\norm{\ketbra{\Psi_{x_0}} - \ketbra{\Psi_0}}_1} \leq \frac{2}{N}\defsum{x_0=1}{N}{\defsum{\ell=0}{K-1}{\norm{(e^{-iH_{x_0}^{(\mathrm{eff})}\Delta t} - e^{-i \langle H_{x_0}^{(\mathrm{eff})}(\ell\Delta t)\rangle \Delta t})\mathcal{T}\prod_{j=1}^{\ell}{V_j\ket{\Tilde{\psi}}}}_2}}.
\end{equation}
For sensor states that only involve classical motion, we can further simplify the right-hand side as
\begin{align}
    \frac{1}{N} \sum_{x_{0} =1} ^{N} 
    \left\|\ketbra{\Psi_{x_{0}}} -  \ketbra{\Psi_{0}} \right  \|_{1}
    &\le \frac{2 K }{N} \max_{\ket{\Tilde{\psi}}  \in \mathcal{H}_{M}^{\mathrm{cl}} }
    \sum_{x_{0} =1} ^{N}  
    \left \| ( e ^{-i H ^{(\text{eff})}_{x_{0}} \Delta t } - e^{-i \mel{\Tilde{\psi}}{H ^{(\text{eff})}_{x_{0}}}{ \Tilde{\psi}} \Delta t} )  \ket{ \Tilde{\psi}   } \right  \|_{2}\\
    &\le \frac{2 K \Delta t }{N} \max_{\ket{\Tilde{\psi}}  \in \mathcal{H}_{M}^{\mathrm{cl}} }
    \sum_{x_{0} =1} ^{N}  
    \sqrt{
    % \braket{ \phi_{ x_{0} } |  (H ^{(\text{eff})}_{x_{0}} ) ^{2}  | \phi_{ x_{0} } }
    \mel{\Tilde{\psi} }{(H ^{(\text{eff})}_{x_{0}} ) ^{2} }{\Tilde{\psi}}
    - 
    % \braket{\phi_{ x_{0} } | H ^{(\text{eff})}_{x_{0}}  | \phi_{ x_{0} }  }
    \mel{\Tilde{\psi}}{H ^{(\text{eff})}_{x_{0}} }{\Tilde{\psi} }
    ^{2}}\\ 
    &\le \frac{4KB\Delta t}{N} M\\
    &= \frac{4MBT}{N},
\end{align}
where we denote by $\mathcal{H}_{M}^{\mathrm{cl}}$ the subset of sensor(-ancilla) states without spatial coherence.
Thus, we obtain the optimal time complexity for protocols without spatial coherence as
\begin{align}
    \label{appeq:T.lb.prod}
    \ket{\tilde{\psi} } \in \mathcal{H}_{M}^{\mathrm{cl}}: \quad T = \Theta\left(\frac{N}{M}\right).
\end{align}
The bound is tight, as we can trivially saturate it by placing $M$ sensors on $M$ sites in each round of the protocol. We note that the resource constraints must be imposed on the joint sensor-ancilla state; otherwise, spatial coherence or entanglement can be hidden in sensor-ancilla correlations. For example, the state $\sum_x\ket{n_{x}=M}\ket{a_x}/\sqrt{N}$ with orthogonal ancilla states $\ket{a_x}$ has a separable reduced sensor state without spatial coherence, yet it achieves the optimal search-time scaling, $T=\Theta(\sqrt{N}/(MB))$.

\begin{thebibliography}{56}%
\makeatletter
\providecommand \@ifxundefined [1]{%
 \@ifx{#1\undefined}
}%
\providecommand \@ifnum [1]{%
 \ifnum #1\expandafter \@firstoftwo
 \else \expandafter \@secondoftwo
 \fi
}%
\providecommand \@ifx [1]{%
 \ifx #1\expandafter \@firstoftwo
 \else \expandafter \@secondoftwo
 \fi
}%
\providecommand \natexlab [1]{#1}%
\providecommand \enquote  [1]{``#1''}%
\providecommand \bibnamefont  [1]{#1}%
\providecommand \bibfnamefont [1]{#1}%
\providecommand \citenamefont [1]{#1}%
\providecommand \href@noop [0]{\@secondoftwo}%
\providecommand \href [0]{\begingroup \@sanitize@url \@href}%
\providecommand \@href[1]{\@@startlink{#1}\@@href}%
\providecommand \@@href[1]{\endgroup#1\@@endlink}%
\providecommand \@sanitize@url [0]{\catcode `\\12\catcode `\$12\catcode
  `\&12\catcode `\#12\catcode `\^12\catcode `\_12\catcode `\%12\relax}%
\providecommand \@@startlink[1]{}%
\providecommand \@@endlink[0]{}%
\providecommand \url  [0]{\begingroup\@sanitize@url \@url }%
\providecommand \@url [1]{\endgroup\@href {#1}{\urlprefix }}%
\providecommand \urlprefix  [0]{URL }%
\providecommand \Eprint [0]{\href }%
\providecommand \doibase [0]{https://doi.org/}%
\providecommand \selectlanguage [0]{\@gobble}%
\providecommand \bibinfo  [0]{\@secondoftwo}%
\providecommand \bibfield  [0]{\@secondoftwo}%
\providecommand \translation [1]{[#1]}%
\providecommand \BibitemOpen [0]{}%
\providecommand \bibitemStop [0]{}%
\providecommand \bibitemNoStop [0]{.\EOS\space}%
\providecommand \EOS [0]{\spacefactor3000\relax}%
\providecommand \BibitemShut  [1]{\csname bibitem#1\endcsname}%
\let\auto@bib@innerbib\@empty
%</preamble>
\bibitem [{\citenamefont {Degen}\ \emph {et~al.}(2017)\citenamefont {Degen},
  \citenamefont {Reinhard},\ and\ \citenamefont {Cappellaro}}]{degen2017}%
  \BibitemOpen
  \bibfield  {author} {\bibinfo {author} {\bibfnamefont {C.}~\bibnamefont
  {Degen}}, \bibinfo {author} {\bibfnamefont {F.}~\bibnamefont {Reinhard}},\
  and\ \bibinfo {author} {\bibfnamefont {P.}~\bibnamefont {Cappellaro}},\
  }\bibfield  {title} {\bibinfo {title} {Quantum sensing},\ }\href
  {https://doi.org/10.1103/RevModPhys.89.035002} {\bibfield  {journal}
  {\bibinfo  {journal} {Reviews of Modern Physics}\ }\textbf {\bibinfo {volume}
  {89}},\ \bibinfo {pages} {035002} (\bibinfo {year} {2017})}\BibitemShut
  {NoStop}%
\bibitem [{\citenamefont {Giovannetti}\ \emph {et~al.}(2011)\citenamefont
  {Giovannetti}, \citenamefont {Lloyd},\ and\ \citenamefont
  {Maccone}}]{Giovannetti2011AdvancesQuantumMetrology}%
  \BibitemOpen
  \bibfield  {author} {\bibinfo {author} {\bibfnamefont {V.}~\bibnamefont
  {Giovannetti}}, \bibinfo {author} {\bibfnamefont {S.}~\bibnamefont {Lloyd}},\
  and\ \bibinfo {author} {\bibfnamefont {L.}~\bibnamefont {Maccone}},\
  }\bibfield  {title} {\bibinfo {title} {Advances in quantum metrology},\
  }\href {https://doi.org/10.1038/nphoton.2011.35} {\bibfield  {journal}
  {\bibinfo  {journal} {Nature Photonics}\ }\textbf {\bibinfo {volume} {5}},\
  \bibinfo {pages} {222} (\bibinfo {year} {2011})}\BibitemShut {NoStop}%
\bibitem [{\citenamefont {Boixo}\ \emph {et~al.}(2007)\citenamefont {Boixo},
  \citenamefont {Flammia}, \citenamefont {Caves},\ and\ \citenamefont
  {Geremia}}]{boixo2007generalized}%
  \BibitemOpen
  \bibfield  {author} {\bibinfo {author} {\bibfnamefont {S.}~\bibnamefont
  {Boixo}}, \bibinfo {author} {\bibfnamefont {S.~T.}\ \bibnamefont {Flammia}},
  \bibinfo {author} {\bibfnamefont {C.~M.}\ \bibnamefont {Caves}},\ and\
  \bibinfo {author} {\bibfnamefont {J.~M.}\ \bibnamefont {Geremia}},\
  }\bibfield  {title} {\bibinfo {title} {Generalized limits for
  single-parameter quantum estimation},\ }\href
  {https://doi.org/10.1103/PhysRevLett.98.090401} {\bibfield  {journal}
  {\bibinfo  {journal} {Physical Review Letters}\ }\textbf {\bibinfo {volume}
  {98}},\ \bibinfo {pages} {090401} (\bibinfo {year} {2007})}\BibitemShut
  {NoStop}%
\bibitem [{\citenamefont {Giovannetti}\ \emph {et~al.}(2006)\citenamefont
  {Giovannetti}, \citenamefont {Lloyd},\ and\ \citenamefont
  {Maccone}}]{Giovannetti2006QuantumMetrology}%
  \BibitemOpen
  \bibfield  {author} {\bibinfo {author} {\bibfnamefont {V.}~\bibnamefont
  {Giovannetti}}, \bibinfo {author} {\bibfnamefont {S.}~\bibnamefont {Lloyd}},\
  and\ \bibinfo {author} {\bibfnamefont {L.}~\bibnamefont {Maccone}},\
  }\bibfield  {title} {\bibinfo {title} {Quantum metrology},\ }\href
  {https://doi.org/10.1103/PhysRevLett.96.010401} {\bibfield  {journal}
  {\bibinfo  {journal} {Physical Review Letters}\ }\textbf {\bibinfo {volume}
  {96}},\ \bibinfo {pages} {010401} (\bibinfo {year} {2006})}\BibitemShut
  {NoStop}%
\bibitem [{\citenamefont {Khan}\ \emph {et~al.}(2025)\citenamefont {Khan},
  \citenamefont {Prabhu}, \citenamefont {Wright},\ and\ \citenamefont
  {McMahon}}]{khan2025quantum}%
  \BibitemOpen
  \bibfield  {author} {\bibinfo {author} {\bibfnamefont {S.~A.}\ \bibnamefont
  {Khan}}, \bibinfo {author} {\bibfnamefont {S.}~\bibnamefont {Prabhu}},
  \bibinfo {author} {\bibfnamefont {L.~G.}\ \bibnamefont {Wright}},\ and\
  \bibinfo {author} {\bibfnamefont {P.~L.}\ \bibnamefont {McMahon}},\
  }\bibfield  {title} {\bibinfo {title} {Quantum computational-sensing
  advantage},\ }\href {https://arxiv.org/abs/2507.16918} {\bibfield  {journal}
  {\bibinfo  {journal} {arXiv preprint arXiv:2507.16918}\ } (\bibinfo {year}
  {2025})}\BibitemShut {NoStop}%
\bibitem [{\citenamefont {Aigner}\ and\ \citenamefont
  {D{\"u}r}(2025)}]{aigner2025power}%
  \BibitemOpen
  \bibfield  {author} {\bibinfo {author} {\bibfnamefont {P.}~\bibnamefont
  {Aigner}}\ and\ \bibinfo {author} {\bibfnamefont {W.}~\bibnamefont
  {D{\"u}r}},\ }\bibfield  {title} {\bibinfo {title} {On the power of moving
  quantum sensors: fully flexible and noise-resilient sensing},\ }\href
  {https://arxiv.org/abs/2512.15876} {\bibfield  {journal} {\bibinfo  {journal}
  {arXiv preprint arXiv:2512.15876}\ } (\bibinfo {year} {2025})}\BibitemShut
  {NoStop}%
\bibitem [{\citenamefont {Eldredge}\ \emph {et~al.}(2018)\citenamefont
  {Eldredge}, \citenamefont {Foss-Feig}, \citenamefont {Gross}, \citenamefont
  {Rolston},\ and\ \citenamefont {Gorshkov}}]{eldredge2018}%
  \BibitemOpen
  \bibfield  {author} {\bibinfo {author} {\bibfnamefont {Z.}~\bibnamefont
  {Eldredge}}, \bibinfo {author} {\bibfnamefont {M.}~\bibnamefont {Foss-Feig}},
  \bibinfo {author} {\bibfnamefont {J.~A.}\ \bibnamefont {Gross}}, \bibinfo
  {author} {\bibfnamefont {S.~L.}\ \bibnamefont {Rolston}},\ and\ \bibinfo
  {author} {\bibfnamefont {A.~V.}\ \bibnamefont {Gorshkov}},\ }\bibfield
  {title} {\bibinfo {title} {Optimal and secure measurement protocols for
  quantum sensor networks},\ }\bibfield  {journal} {\bibinfo  {journal}
  {Physical Review A}\ }\textbf {\bibinfo {volume} {97}},\ \href
  {https://doi.org/10.1103/PhysRevA.97.042337} {10.1103/PhysRevA.97.042337}
  (\bibinfo {year} {2018})\BibitemShut {NoStop}%
\bibitem [{\citenamefont {Abbasgholinejad}\ \emph {et~al.}(2025)\citenamefont
  {Abbasgholinejad}, \citenamefont {Muleady}, \citenamefont {Bringewatt},
  \citenamefont {Brady}, \citenamefont {Wang}, \citenamefont {Fahimniya},\ and\
  \citenamefont {Gorshkov}}]{abbasgholinejad2025optimally}%
  \BibitemOpen
  \bibfield  {author} {\bibinfo {author} {\bibfnamefont {E.}~\bibnamefont
  {Abbasgholinejad}}, \bibinfo {author} {\bibfnamefont {S.~R.}\ \bibnamefont
  {Muleady}}, \bibinfo {author} {\bibfnamefont {J.}~\bibnamefont {Bringewatt}},
  \bibinfo {author} {\bibfnamefont {A.~J.}\ \bibnamefont {Brady}}, \bibinfo
  {author} {\bibfnamefont {Y.-X.}\ \bibnamefont {Wang}}, \bibinfo {author}
  {\bibfnamefont {A.}~\bibnamefont {Fahimniya}},\ and\ \bibinfo {author}
  {\bibfnamefont {A.~V.}\ \bibnamefont {Gorshkov}},\ }\bibfield  {title}
  {\bibinfo {title} {Optimally learning functions in interacting quantum sensor
  networks},\ }\href {https://arxiv.org/abs/2510.06360} {\bibfield  {journal}
  {\bibinfo  {journal} {arXiv preprint arXiv:2510.06360}\ } (\bibinfo {year}
  {2025})}\BibitemShut {NoStop}%
\bibitem [{\citenamefont {Maletinsky}\ \emph {et~al.}(2012)\citenamefont
  {Maletinsky}, \citenamefont {Hong}, \citenamefont {Grinolds}, \citenamefont
  {Hausmann}, \citenamefont {Lukin}, \citenamefont {Walsworth}, \citenamefont
  {Lon{\v{c}}ar},\ and\ \citenamefont {Yacoby}}]{Maletinsky2012ScanningNV}%
  \BibitemOpen
  \bibfield  {author} {\bibinfo {author} {\bibfnamefont {P.}~\bibnamefont
  {Maletinsky}}, \bibinfo {author} {\bibfnamefont {S.}~\bibnamefont {Hong}},
  \bibinfo {author} {\bibfnamefont {M.~S.}\ \bibnamefont {Grinolds}}, \bibinfo
  {author} {\bibfnamefont {B.}~\bibnamefont {Hausmann}}, \bibinfo {author}
  {\bibfnamefont {M.~D.}\ \bibnamefont {Lukin}}, \bibinfo {author}
  {\bibfnamefont {R.~L.}\ \bibnamefont {Walsworth}}, \bibinfo {author}
  {\bibfnamefont {M.}~\bibnamefont {Lon{\v{c}}ar}},\ and\ \bibinfo {author}
  {\bibfnamefont {A.}~\bibnamefont {Yacoby}},\ }\bibfield  {title} {\bibinfo
  {title} {A robust scanning diamond sensor for nanoscale imaging with single
  nitrogen-vacancy centres},\ }\href {https://doi.org/10.1038/nnano.2012.50}
  {\bibfield  {journal} {\bibinfo  {journal} {Nature Nanotechnology}\ }\textbf
  {\bibinfo {volume} {7}},\ \bibinfo {pages} {320} (\bibinfo {year}
  {2012})}\BibitemShut {NoStop}%
\bibitem [{\citenamefont {Schirhagl}\ \emph {et~al.}(2014)\citenamefont
  {Schirhagl}, \citenamefont {Chang}, \citenamefont {Loretz},\ and\
  \citenamefont {Degen}}]{Schirhagl2014NVReview}%
  \BibitemOpen
  \bibfield  {author} {\bibinfo {author} {\bibfnamefont {R.}~\bibnamefont
  {Schirhagl}}, \bibinfo {author} {\bibfnamefont {K.}~\bibnamefont {Chang}},
  \bibinfo {author} {\bibfnamefont {M.}~\bibnamefont {Loretz}},\ and\ \bibinfo
  {author} {\bibfnamefont {C.~L.}\ \bibnamefont {Degen}},\ }\bibfield  {title}
  {\bibinfo {title} {Nitrogen-vacancy centers in diamond: Nanoscale sensors for
  physics and biology},\ }\href
  {https://doi.org/10.1146/annurev-physchem-040513-103659} {\bibfield
  {journal} {\bibinfo  {journal} {Annual Review of Physical Chemistry}\
  }\textbf {\bibinfo {volume} {65}},\ \bibinfo {pages} {83} (\bibinfo {year}
  {2014})}\BibitemShut {NoStop}%
\bibitem [{\citenamefont {Rondin}\ \emph {et~al.}(2014)\citenamefont {Rondin},
  \citenamefont {T{\'e}tienne}, \citenamefont {Hingant}, \citenamefont {Roch},
  \citenamefont {Maletinsky},\ and\ \citenamefont
  {Jacques}}]{Rondin2014NVReview}%
  \BibitemOpen
  \bibfield  {author} {\bibinfo {author} {\bibfnamefont {L.}~\bibnamefont
  {Rondin}}, \bibinfo {author} {\bibfnamefont {J.-P.}\ \bibnamefont
  {T{\'e}tienne}}, \bibinfo {author} {\bibfnamefont {T.}~\bibnamefont
  {Hingant}}, \bibinfo {author} {\bibfnamefont {J.-F.}\ \bibnamefont {Roch}},
  \bibinfo {author} {\bibfnamefont {P.}~\bibnamefont {Maletinsky}},\ and\
  \bibinfo {author} {\bibfnamefont {V.}~\bibnamefont {Jacques}},\ }\bibfield
  {title} {\bibinfo {title} {Magnetometry with nitrogen-vacancy defects in
  diamond},\ }\href {https://doi.org/10.1088/0034-4885/77/5/056503} {\bibfield
  {journal} {\bibinfo  {journal} {Reports on Progress in Physics}\ }\textbf
  {\bibinfo {volume} {77}},\ \bibinfo {pages} {056503} (\bibinfo {year}
  {2014})}\BibitemShut {NoStop}%
\bibitem [{\citenamefont {Kielpinski}\ \emph {et~al.}(2002)\citenamefont
  {Kielpinski}, \citenamefont {Monroe},\ and\ \citenamefont
  {Wineland}}]{Kielpinski2002IonTrapArchitecture}%
  \BibitemOpen
  \bibfield  {author} {\bibinfo {author} {\bibfnamefont {D.}~\bibnamefont
  {Kielpinski}}, \bibinfo {author} {\bibfnamefont {C.}~\bibnamefont {Monroe}},\
  and\ \bibinfo {author} {\bibfnamefont {D.~J.}\ \bibnamefont {Wineland}},\
  }\bibfield  {title} {\bibinfo {title} {Architecture for a large-scale
  ion-trap quantum computer},\ }\href {https://doi.org/10.1038/nature00784}
  {\bibfield  {journal} {\bibinfo  {journal} {Nature}\ }\textbf {\bibinfo
  {volume} {417}},\ \bibinfo {pages} {709} (\bibinfo {year}
  {2002})}\BibitemShut {NoStop}%
\bibitem [{\citenamefont {Blakestad}\ \emph {et~al.}(2009)\citenamefont
  {Blakestad}, \citenamefont {Ospelkaus}, \citenamefont {VanDevender},
  \citenamefont {Amini}, \citenamefont {Britton}, \citenamefont {Leibfried},\
  and\ \citenamefont {Wineland}}]{Blakestad2009XJunction}%
  \BibitemOpen
  \bibfield  {author} {\bibinfo {author} {\bibfnamefont {R.~B.}\ \bibnamefont
  {Blakestad}}, \bibinfo {author} {\bibfnamefont {C.}~\bibnamefont
  {Ospelkaus}}, \bibinfo {author} {\bibfnamefont {A.~P.}\ \bibnamefont
  {VanDevender}}, \bibinfo {author} {\bibfnamefont {J.~M.}\ \bibnamefont
  {Amini}}, \bibinfo {author} {\bibfnamefont {J.}~\bibnamefont {Britton}},
  \bibinfo {author} {\bibfnamefont {D.}~\bibnamefont {Leibfried}},\ and\
  \bibinfo {author} {\bibfnamefont {D.~J.}\ \bibnamefont {Wineland}},\
  }\bibfield  {title} {\bibinfo {title} {High-fidelity transport of trapped-ion
  qubits through an x-junction trap array},\ }\href
  {https://doi.org/10.1103/PhysRevLett.102.153002} {\bibfield  {journal}
  {\bibinfo  {journal} {Physical Review Letters}\ }\textbf {\bibinfo {volume}
  {102}},\ \bibinfo {pages} {153002} (\bibinfo {year} {2009})}\BibitemShut
  {NoStop}%
\bibitem [{\citenamefont {Bowler}\ \emph {et~al.}(2012)\citenamefont {Bowler},
  \citenamefont {Gaebler}, \citenamefont {Lin}, \citenamefont {Tan},
  \citenamefont {Hanneke}, \citenamefont {Jost}, \citenamefont {Home},
  \citenamefont {Leibfried},\ and\ \citenamefont
  {Wineland}}]{Bowler2012DiabaticTransport}%
  \BibitemOpen
  \bibfield  {author} {\bibinfo {author} {\bibfnamefont {R.}~\bibnamefont
  {Bowler}}, \bibinfo {author} {\bibfnamefont {J.}~\bibnamefont {Gaebler}},
  \bibinfo {author} {\bibfnamefont {Y.}~\bibnamefont {Lin}}, \bibinfo {author}
  {\bibfnamefont {T.~R.}\ \bibnamefont {Tan}}, \bibinfo {author} {\bibfnamefont
  {D.}~\bibnamefont {Hanneke}}, \bibinfo {author} {\bibfnamefont {J.~D.}\
  \bibnamefont {Jost}}, \bibinfo {author} {\bibfnamefont {J.~P.}\ \bibnamefont
  {Home}}, \bibinfo {author} {\bibfnamefont {D.}~\bibnamefont {Leibfried}},\
  and\ \bibinfo {author} {\bibfnamefont {D.~J.}\ \bibnamefont {Wineland}},\
  }\bibfield  {title} {\bibinfo {title} {Coherent diabatic ion transport and
  separation in a multizone trap array},\ }\href
  {https://doi.org/10.1103/PhysRevLett.109.080502} {\bibfield  {journal}
  {\bibinfo  {journal} {Physical Review Letters}\ }\textbf {\bibinfo {volume}
  {109}},\ \bibinfo {pages} {080502} (\bibinfo {year} {2012})}\BibitemShut
  {NoStop}%
\bibitem [{\citenamefont {Parazzoli}\ \emph {et~al.}(2022)\citenamefont
  {Parazzoli}, \citenamefont {Van Der~Wall},\ and\ \citenamefont
  {Stick}}]{Parazzoli2022Shuttling}%
  \BibitemOpen
  \bibfield  {author} {\bibinfo {author} {\bibfnamefont {P.}~\bibnamefont
  {Parazzoli}}, \bibinfo {author} {\bibfnamefont {J.}~\bibnamefont {Van
  Der~Wall}},\ and\ \bibinfo {author} {\bibfnamefont {D.}~\bibnamefont
  {Stick}},\ }\bibfield  {title} {\bibinfo {title} {Closed-loop optimization of
  fast trapped-ion shuttling with sub-quanta excitation},\ }\href
  {https://doi.org/10.1038/s41534-022-00579-3} {\bibfield  {journal} {\bibinfo
  {journal} {npj Quantum Information}\ }\textbf {\bibinfo {volume} {8}},\
  \bibinfo {pages} {68} (\bibinfo {year} {2022})}\BibitemShut {NoStop}%
\bibitem [{\citenamefont {Ransford}\ \emph {et~al.}(2026)\citenamefont
  {Ransford}, \citenamefont {Allman}, \citenamefont {Arkinstall}, \citenamefont
  {Campora}, \citenamefont {Cooper}, \citenamefont {Delaney}, \citenamefont
  {Dreiling}, \citenamefont {Estey}, \citenamefont {Figgatt}, \citenamefont
  {Hall}, \citenamefont {Husain}, \citenamefont {Isanaka}, \citenamefont
  {Kennedy}, \citenamefont {Kotibhaskar}, \citenamefont {Madjarov},
  \citenamefont {Mayer}, \citenamefont {Milne}, \citenamefont {Park},
  \citenamefont {Reed}, \citenamefont {Ancona}, \citenamefont {Andersen},
  \citenamefont {{Andres-Martinez}}, \citenamefont {Angenent}, \citenamefont
  {Argueta}, \citenamefont {Arkin}, \citenamefont {Ascarrunz}, \citenamefont
  {Baker}, \citenamefont {Barnes}, \citenamefont {Bartolotta}, \citenamefont
  {Berg}, \citenamefont {Besand}, \citenamefont {Bjork}, \citenamefont {Blain},
  \citenamefont {Blanchard}, \citenamefont {{Blume-Kohout}}, \citenamefont
  {Bohn}, \citenamefont {Borgna}, \citenamefont {Botamanenko}, \citenamefont
  {Boutelle}, \citenamefont {Brown}, \citenamefont {Buckingham}, \citenamefont
  {Burdick}, \citenamefont {Burton}, \citenamefont {Carey}, \citenamefont
  {Carron}, \citenamefont {Chambers}, \citenamefont {Chan}, \citenamefont
  {Children}, \citenamefont {Colussi}, \citenamefont {Crepinsek}, \citenamefont
  {Cureton}, \citenamefont {Davies}, \citenamefont {Davis}, \citenamefont
  {DeCross}, \citenamefont {Deen}, \citenamefont {Delaney}, \citenamefont
  {DelVento}, \citenamefont {DeSalvo}, \citenamefont {Dominy}, \citenamefont
  {Drotar}, \citenamefont {Duncan}, \citenamefont {Eccles}, \citenamefont
  {Edgington}, \citenamefont {Erickson}, \citenamefont {Erickson},
  \citenamefont {Ertsgaard}, \citenamefont {Esposito}, \citenamefont {Evans},
  \citenamefont {Evans}, \citenamefont {Fabrikant}, \citenamefont {Fischer},
  \citenamefont {Foltz}, \citenamefont {{Foss-Feig}}, \citenamefont {Francois},
  \citenamefont {Freyberg}, \citenamefont {Gao}, \citenamefont {Garay},
  \citenamefont {Garvin}, \citenamefont {Gaudiosi}, \citenamefont {Gilbreth},
  \citenamefont {Giles}, \citenamefont {Glynn}, \citenamefont {Graves},
  \citenamefont {Hansen}, \citenamefont {Hayes}, \citenamefont {Heidemann},
  \citenamefont {Higashi}, \citenamefont {Hilbun}, \citenamefont {Hines},
  \citenamefont {Hlavaty}, \citenamefont {Hoffman}, \citenamefont {Hoffman},
  \citenamefont {Holliman}, \citenamefont {Hooper}, \citenamefont {Horning},
  \citenamefont {Hostetter}, \citenamefont {Hothem}, \citenamefont {Houlton},
  \citenamefont {Hout}, \citenamefont {Hutson}, \citenamefont {Jacobs},
  \citenamefont {Jacobs}, \citenamefont {Johannsen}, \citenamefont {Johansen},
  \citenamefont {Jones}, \citenamefont {Julian}, \citenamefont {Jung},
  \citenamefont {Keay}, \citenamefont {Klein}, \citenamefont {Koch},
  \citenamefont {Kondo}, \citenamefont {Kong}, \citenamefont {Kosto},
  \citenamefont {Lawrence}, \citenamefont {Liefer}, \citenamefont {Lollie},
  \citenamefont {Lucchetti}, \citenamefont {Lysne}, \citenamefont {Lytle},
  \citenamefont {MacPherson}, \citenamefont {Malm}, \citenamefont {Mather},
  \citenamefont {Mathewson}, \citenamefont {Maxwell}, \citenamefont
  {McCaffrey}, \citenamefont {McDougall}, \citenamefont {Mendoza},
  \citenamefont {Miller}, \citenamefont {Mills}, \citenamefont {Morrison},
  \citenamefont {Narmour}, \citenamefont {Nguyen}, \citenamefont {Nugent},
  \citenamefont {Olson}, \citenamefont {Ouellette}, \citenamefont {Parks},
  \citenamefont {Peters}, \citenamefont {Peterson}, \citenamefont {Petricka},
  \citenamefont {Pino}, \citenamefont {Polito}, \citenamefont {Potter},
  \citenamefont {Preidl}, \citenamefont {Price}, \citenamefont {Proctor},
  \citenamefont {Pugh}, \citenamefont {Ratcliff}, \citenamefont {Raymondson},
  \citenamefont {Rhodes}, \citenamefont {Roman}, \citenamefont {Roy},
  \citenamefont {{Ryan-Anderson}}, \citenamefont {Sanchez}, \citenamefont
  {Sangiolo}, \citenamefont {Sawadski}, \citenamefont {Schaffer}, \citenamefont
  {Schow}, \citenamefont {Sedlacek}, \citenamefont {Semenenko}, \citenamefont
  {Shevchuk}, \citenamefont {Shore}, \citenamefont {Siegfried}, \citenamefont
  {Singhal}, \citenamefont {Sivarajah}, \citenamefont {Skripka}, \citenamefont
  {Sletten}, \citenamefont {Spaun}, \citenamefont {Sprenkle}, \citenamefont
  {Stoufer}, \citenamefont {Tader}, \citenamefont {Taylor}, \citenamefont
  {Thompson}, \citenamefont {Tobey}, \citenamefont {Tran}, \citenamefont
  {Tran}, \citenamefont {Vittorini}, \citenamefont {Volin}, \citenamefont
  {Walker}, \citenamefont {White}, \citenamefont {Williams}, \citenamefont
  {Wilson}, \citenamefont {Wolf}, \citenamefont {Wringe}, \citenamefont
  {Young}, \citenamefont {Zheng}, \citenamefont {Zuraski}, \citenamefont
  {Baldwin}, \citenamefont {Chernoguzov}, \citenamefont {Gaebler},
  \citenamefont {Sanders}, \citenamefont {Neyenhuis}, \citenamefont {Stutz},\
  and\ \citenamefont {Bohnet}}]{ransford_98-qubit_2026}%
  \BibitemOpen
  \bibfield  {author} {\bibinfo {author} {\bibfnamefont {A.}~\bibnamefont
  {Ransford}}, \bibinfo {author} {\bibfnamefont {M.~S.}\ \bibnamefont
  {Allman}}, \bibinfo {author} {\bibfnamefont {J.}~\bibnamefont {Arkinstall}},
  \bibinfo {author} {\bibfnamefont {J.~P.}\ \bibnamefont {Campora}}, \bibinfo
  {author} {\bibfnamefont {S.~F.}\ \bibnamefont {Cooper}}, \bibinfo {author}
  {\bibfnamefont {R.~D.}\ \bibnamefont {Delaney}}, \bibinfo {author}
  {\bibfnamefont {J.~M.}\ \bibnamefont {Dreiling}}, \bibinfo {author}
  {\bibfnamefont {B.}~\bibnamefont {Estey}}, \bibinfo {author} {\bibfnamefont
  {C.}~\bibnamefont {Figgatt}}, \bibinfo {author} {\bibfnamefont
  {A.}~\bibnamefont {Hall}}, \bibinfo {author} {\bibfnamefont {A.~A.}\
  \bibnamefont {Husain}}, \bibinfo {author} {\bibfnamefont {A.}~\bibnamefont
  {Isanaka}}, \bibinfo {author} {\bibfnamefont {C.~J.}\ \bibnamefont
  {Kennedy}}, \bibinfo {author} {\bibfnamefont {N.}~\bibnamefont
  {Kotibhaskar}}, \bibinfo {author} {\bibfnamefont {I.~S.}\ \bibnamefont
  {Madjarov}}, \bibinfo {author} {\bibfnamefont {K.}~\bibnamefont {Mayer}},
  \bibinfo {author} {\bibfnamefont {A.~R.}\ \bibnamefont {Milne}}, \bibinfo
  {author} {\bibfnamefont {A.~J.}\ \bibnamefont {Park}}, \bibinfo {author}
  {\bibfnamefont {A.~P.}\ \bibnamefont {Reed}}, \bibinfo {author}
  {\bibfnamefont {R.}~\bibnamefont {Ancona}}, \bibinfo {author} {\bibfnamefont
  {M.~P.}\ \bibnamefont {Andersen}}, \bibinfo {author} {\bibfnamefont
  {P.}~\bibnamefont {{Andres-Martinez}}}, \bibinfo {author} {\bibfnamefont
  {W.}~\bibnamefont {Angenent}}, \bibinfo {author} {\bibfnamefont
  {L.}~\bibnamefont {Argueta}}, \bibinfo {author} {\bibfnamefont
  {B.}~\bibnamefont {Arkin}}, \bibinfo {author} {\bibfnamefont
  {L.}~\bibnamefont {Ascarrunz}}, \bibinfo {author} {\bibfnamefont
  {W.}~\bibnamefont {Baker}}, \bibinfo {author} {\bibfnamefont
  {C.}~\bibnamefont {Barnes}}, \bibinfo {author} {\bibfnamefont
  {J.}~\bibnamefont {Bartolotta}}, \bibinfo {author} {\bibfnamefont
  {J.}~\bibnamefont {Berg}}, \bibinfo {author} {\bibfnamefont {R.}~\bibnamefont
  {Besand}}, \bibinfo {author} {\bibfnamefont {B.}~\bibnamefont {Bjork}},
  \bibinfo {author} {\bibfnamefont {M.}~\bibnamefont {Blain}}, \bibinfo
  {author} {\bibfnamefont {P.}~\bibnamefont {Blanchard}}, \bibinfo {author}
  {\bibfnamefont {R.}~\bibnamefont {{Blume-Kohout}}}, \bibinfo {author}
  {\bibfnamefont {M.}~\bibnamefont {Bohn}}, \bibinfo {author} {\bibfnamefont
  {A.}~\bibnamefont {Borgna}}, \bibinfo {author} {\bibfnamefont {D.~Y.}\
  \bibnamefont {Botamanenko}}, \bibinfo {author} {\bibfnamefont
  {R.}~\bibnamefont {Boutelle}}, \bibinfo {author} {\bibfnamefont
  {N.}~\bibnamefont {Brown}}, \bibinfo {author} {\bibfnamefont {G.~T.}\
  \bibnamefont {Buckingham}}, \bibinfo {author} {\bibfnamefont {N.~Q.}\
  \bibnamefont {Burdick}}, \bibinfo {author} {\bibfnamefont {W.~C.}\
  \bibnamefont {Burton}}, \bibinfo {author} {\bibfnamefont {V.}~\bibnamefont
  {Carey}}, \bibinfo {author} {\bibfnamefont {C.~J.}\ \bibnamefont {Carron}},
  \bibinfo {author} {\bibfnamefont {J.}~\bibnamefont {Chambers}}, \bibinfo
  {author} {\bibfnamefont {J.~W.}\ \bibnamefont {Chan}}, \bibinfo {author}
  {\bibfnamefont {J.}~\bibnamefont {Children}}, \bibinfo {author}
  {\bibfnamefont {V.~E.}\ \bibnamefont {Colussi}}, \bibinfo {author}
  {\bibfnamefont {S.}~\bibnamefont {Crepinsek}}, \bibinfo {author}
  {\bibfnamefont {A.}~\bibnamefont {Cureton}}, \bibinfo {author} {\bibfnamefont
  {J.}~\bibnamefont {Davies}}, \bibinfo {author} {\bibfnamefont
  {D.}~\bibnamefont {Davis}}, \bibinfo {author} {\bibfnamefont
  {M.}~\bibnamefont {DeCross}}, \bibinfo {author} {\bibfnamefont
  {D.}~\bibnamefont {Deen}}, \bibinfo {author} {\bibfnamefont {C.}~\bibnamefont
  {Delaney}}, \bibinfo {author} {\bibfnamefont {D.}~\bibnamefont {DelVento}},
  \bibinfo {author} {\bibfnamefont {B.~J.}\ \bibnamefont {DeSalvo}}, \bibinfo
  {author} {\bibfnamefont {J.}~\bibnamefont {Dominy}}, \bibinfo {author}
  {\bibfnamefont {S.}~\bibnamefont {Drotar}}, \bibinfo {author} {\bibfnamefont
  {R.}~\bibnamefont {Duncan}}, \bibinfo {author} {\bibfnamefont
  {V.}~\bibnamefont {Eccles}}, \bibinfo {author} {\bibfnamefont
  {A.}~\bibnamefont {Edgington}}, \bibinfo {author} {\bibfnamefont
  {N.}~\bibnamefont {Erickson}}, \bibinfo {author} {\bibfnamefont
  {S.}~\bibnamefont {Erickson}}, \bibinfo {author} {\bibfnamefont {C.~T.}\
  \bibnamefont {Ertsgaard}}, \bibinfo {author} {\bibfnamefont {J.}~\bibnamefont
  {Esposito}}, \bibinfo {author} {\bibfnamefont {B.}~\bibnamefont {Evans}},
  \bibinfo {author} {\bibfnamefont {T.}~\bibnamefont {Evans}}, \bibinfo
  {author} {\bibfnamefont {M.~I.}\ \bibnamefont {Fabrikant}}, \bibinfo {author}
  {\bibfnamefont {A.}~\bibnamefont {Fischer}}, \bibinfo {author} {\bibfnamefont
  {C.}~\bibnamefont {Foltz}}, \bibinfo {author} {\bibfnamefont
  {M.}~\bibnamefont {{Foss-Feig}}}, \bibinfo {author} {\bibfnamefont
  {D.}~\bibnamefont {Francois}}, \bibinfo {author} {\bibfnamefont
  {B.}~\bibnamefont {Freyberg}}, \bibinfo {author} {\bibfnamefont
  {C.}~\bibnamefont {Gao}}, \bibinfo {author} {\bibfnamefont {R.}~\bibnamefont
  {Garay}}, \bibinfo {author} {\bibfnamefont {J.}~\bibnamefont {Garvin}},
  \bibinfo {author} {\bibfnamefont {D.~M.}\ \bibnamefont {Gaudiosi}}, \bibinfo
  {author} {\bibfnamefont {C.~N.}\ \bibnamefont {Gilbreth}}, \bibinfo {author}
  {\bibfnamefont {J.}~\bibnamefont {Giles}}, \bibinfo {author} {\bibfnamefont
  {E.}~\bibnamefont {Glynn}}, \bibinfo {author} {\bibfnamefont
  {J.}~\bibnamefont {Graves}}, \bibinfo {author} {\bibfnamefont
  {A.}~\bibnamefont {Hansen}}, \bibinfo {author} {\bibfnamefont
  {D.}~\bibnamefont {Hayes}}, \bibinfo {author} {\bibfnamefont
  {L.}~\bibnamefont {Heidemann}}, \bibinfo {author} {\bibfnamefont
  {B.}~\bibnamefont {Higashi}}, \bibinfo {author} {\bibfnamefont
  {T.}~\bibnamefont {Hilbun}}, \bibinfo {author} {\bibfnamefont
  {J.}~\bibnamefont {Hines}}, \bibinfo {author} {\bibfnamefont
  {A.}~\bibnamefont {Hlavaty}}, \bibinfo {author} {\bibfnamefont
  {K.}~\bibnamefont {Hoffman}}, \bibinfo {author} {\bibfnamefont {I.~M.}\
  \bibnamefont {Hoffman}}, \bibinfo {author} {\bibfnamefont {C.}~\bibnamefont
  {Holliman}}, \bibinfo {author} {\bibfnamefont {I.}~\bibnamefont {Hooper}},
  \bibinfo {author} {\bibfnamefont {B.}~\bibnamefont {Horning}}, \bibinfo
  {author} {\bibfnamefont {J.}~\bibnamefont {Hostetter}}, \bibinfo {author}
  {\bibfnamefont {D.}~\bibnamefont {Hothem}}, \bibinfo {author} {\bibfnamefont
  {J.}~\bibnamefont {Houlton}}, \bibinfo {author} {\bibfnamefont
  {J.}~\bibnamefont {Hout}}, \bibinfo {author} {\bibfnamefont {R.}~\bibnamefont
  {Hutson}}, \bibinfo {author} {\bibfnamefont {R.~T.}\ \bibnamefont {Jacobs}},
  \bibinfo {author} {\bibfnamefont {T.}~\bibnamefont {Jacobs}}, \bibinfo
  {author} {\bibfnamefont {M.}~\bibnamefont {Johannsen}}, \bibinfo {author}
  {\bibfnamefont {J.}~\bibnamefont {Johansen}}, \bibinfo {author}
  {\bibfnamefont {L.}~\bibnamefont {Jones}}, \bibinfo {author} {\bibfnamefont
  {S.}~\bibnamefont {Julian}}, \bibinfo {author} {\bibfnamefont
  {R.}~\bibnamefont {Jung}}, \bibinfo {author} {\bibfnamefont {A.}~\bibnamefont
  {Keay}}, \bibinfo {author} {\bibfnamefont {T.}~\bibnamefont {Klein}},
  \bibinfo {author} {\bibfnamefont {M.}~\bibnamefont {Koch}}, \bibinfo {author}
  {\bibfnamefont {R.}~\bibnamefont {Kondo}}, \bibinfo {author} {\bibfnamefont
  {C.}~\bibnamefont {Kong}}, \bibinfo {author} {\bibfnamefont {A.}~\bibnamefont
  {Kosto}}, \bibinfo {author} {\bibfnamefont {A.}~\bibnamefont {Lawrence}},
  \bibinfo {author} {\bibfnamefont {D.}~\bibnamefont {Liefer}}, \bibinfo
  {author} {\bibfnamefont {M.}~\bibnamefont {Lollie}}, \bibinfo {author}
  {\bibfnamefont {D.}~\bibnamefont {Lucchetti}}, \bibinfo {author}
  {\bibfnamefont {N.~K.}\ \bibnamefont {Lysne}}, \bibinfo {author}
  {\bibfnamefont {C.}~\bibnamefont {Lytle}}, \bibinfo {author} {\bibfnamefont
  {C.}~\bibnamefont {MacPherson}}, \bibinfo {author} {\bibfnamefont
  {A.}~\bibnamefont {Malm}}, \bibinfo {author} {\bibfnamefont {S.}~\bibnamefont
  {Mather}}, \bibinfo {author} {\bibfnamefont {B.}~\bibnamefont {Mathewson}},
  \bibinfo {author} {\bibfnamefont {D.}~\bibnamefont {Maxwell}}, \bibinfo
  {author} {\bibfnamefont {L.}~\bibnamefont {McCaffrey}}, \bibinfo {author}
  {\bibfnamefont {H.}~\bibnamefont {McDougall}}, \bibinfo {author}
  {\bibfnamefont {R.}~\bibnamefont {Mendoza}}, \bibinfo {author} {\bibfnamefont
  {D.~B.}\ \bibnamefont {Miller}}, \bibinfo {author} {\bibfnamefont
  {M.}~\bibnamefont {Mills}}, \bibinfo {author} {\bibfnamefont
  {R.}~\bibnamefont {Morrison}}, \bibinfo {author} {\bibfnamefont
  {L.}~\bibnamefont {Narmour}}, \bibinfo {author} {\bibfnamefont
  {N.}~\bibnamefont {Nguyen}}, \bibinfo {author} {\bibfnamefont
  {L.}~\bibnamefont {Nugent}}, \bibinfo {author} {\bibfnamefont
  {S.}~\bibnamefont {Olson}}, \bibinfo {author} {\bibfnamefont
  {D.}~\bibnamefont {Ouellette}}, \bibinfo {author} {\bibfnamefont
  {J.}~\bibnamefont {Parks}}, \bibinfo {author} {\bibfnamefont
  {Z.}~\bibnamefont {Peters}}, \bibinfo {author} {\bibfnamefont {T.~A.}\
  \bibnamefont {Peterson}}, \bibinfo {author} {\bibfnamefont {J.}~\bibnamefont
  {Petricka}}, \bibinfo {author} {\bibfnamefont {J.~M.}\ \bibnamefont {Pino}},
  \bibinfo {author} {\bibfnamefont {F.}~\bibnamefont {Polito}}, \bibinfo
  {author} {\bibfnamefont {A.~C.}\ \bibnamefont {Potter}}, \bibinfo {author}
  {\bibfnamefont {M.}~\bibnamefont {Preidl}}, \bibinfo {author} {\bibfnamefont
  {G.}~\bibnamefont {Price}}, \bibinfo {author} {\bibfnamefont
  {T.}~\bibnamefont {Proctor}}, \bibinfo {author} {\bibfnamefont
  {M.}~\bibnamefont {Pugh}}, \bibinfo {author} {\bibfnamefont {N.}~\bibnamefont
  {Ratcliff}}, \bibinfo {author} {\bibfnamefont {D.}~\bibnamefont
  {Raymondson}}, \bibinfo {author} {\bibfnamefont {P.}~\bibnamefont {Rhodes}},
  \bibinfo {author} {\bibfnamefont {C.}~\bibnamefont {Roman}}, \bibinfo
  {author} {\bibfnamefont {C.}~\bibnamefont {Roy}}, \bibinfo {author}
  {\bibfnamefont {C.}~\bibnamefont {{Ryan-Anderson}}}, \bibinfo {author}
  {\bibfnamefont {F.~B.}\ \bibnamefont {Sanchez}}, \bibinfo {author}
  {\bibfnamefont {G.}~\bibnamefont {Sangiolo}}, \bibinfo {author}
  {\bibfnamefont {T.}~\bibnamefont {Sawadski}}, \bibinfo {author}
  {\bibfnamefont {A.}~\bibnamefont {Schaffer}}, \bibinfo {author}
  {\bibfnamefont {P.}~\bibnamefont {Schow}}, \bibinfo {author} {\bibfnamefont
  {J.}~\bibnamefont {Sedlacek}}, \bibinfo {author} {\bibfnamefont
  {H.}~\bibnamefont {Semenenko}}, \bibinfo {author} {\bibfnamefont
  {P.}~\bibnamefont {Shevchuk}}, \bibinfo {author} {\bibfnamefont
  {S.}~\bibnamefont {Shore}}, \bibinfo {author} {\bibfnamefont
  {P.}~\bibnamefont {Siegfried}}, \bibinfo {author} {\bibfnamefont
  {K.}~\bibnamefont {Singhal}}, \bibinfo {author} {\bibfnamefont
  {S.}~\bibnamefont {Sivarajah}}, \bibinfo {author} {\bibfnamefont
  {T.}~\bibnamefont {Skripka}}, \bibinfo {author} {\bibfnamefont
  {L.}~\bibnamefont {Sletten}}, \bibinfo {author} {\bibfnamefont
  {B.}~\bibnamefont {Spaun}}, \bibinfo {author} {\bibfnamefont {R.~T.}\
  \bibnamefont {Sprenkle}}, \bibinfo {author} {\bibfnamefont {P.}~\bibnamefont
  {Stoufer}}, \bibinfo {author} {\bibfnamefont {M.}~\bibnamefont {Tader}},
  \bibinfo {author} {\bibfnamefont {S.~F.}\ \bibnamefont {Taylor}}, \bibinfo
  {author} {\bibfnamefont {T.~H.}\ \bibnamefont {Thompson}}, \bibinfo {author}
  {\bibfnamefont {R.}~\bibnamefont {Tobey}}, \bibinfo {author} {\bibfnamefont
  {A.}~\bibnamefont {Tran}}, \bibinfo {author} {\bibfnamefont {T.}~\bibnamefont
  {Tran}}, \bibinfo {author} {\bibfnamefont {G.}~\bibnamefont {Vittorini}},
  \bibinfo {author} {\bibfnamefont {C.}~\bibnamefont {Volin}}, \bibinfo
  {author} {\bibfnamefont {J.}~\bibnamefont {Walker}}, \bibinfo {author}
  {\bibfnamefont {S.}~\bibnamefont {White}}, \bibinfo {author} {\bibfnamefont
  {G.~R.}\ \bibnamefont {Williams}}, \bibinfo {author} {\bibfnamefont
  {D.}~\bibnamefont {Wilson}}, \bibinfo {author} {\bibfnamefont
  {Q.}~\bibnamefont {Wolf}}, \bibinfo {author} {\bibfnamefont {C.}~\bibnamefont
  {Wringe}}, \bibinfo {author} {\bibfnamefont {K.}~\bibnamefont {Young}},
  \bibinfo {author} {\bibfnamefont {J.}~\bibnamefont {Zheng}}, \bibinfo
  {author} {\bibfnamefont {K.}~\bibnamefont {Zuraski}}, \bibinfo {author}
  {\bibfnamefont {C.~H.}\ \bibnamefont {Baldwin}}, \bibinfo {author}
  {\bibfnamefont {A.}~\bibnamefont {Chernoguzov}}, \bibinfo {author}
  {\bibfnamefont {J.~P.}\ \bibnamefont {Gaebler}}, \bibinfo {author}
  {\bibfnamefont {S.~J.}\ \bibnamefont {Sanders}}, \bibinfo {author}
  {\bibfnamefont {B.}~\bibnamefont {Neyenhuis}}, \bibinfo {author}
  {\bibfnamefont {R.}~\bibnamefont {Stutz}},\ and\ \bibinfo {author}
  {\bibfnamefont {J.~G.}\ \bibnamefont {Bohnet}},\ }\bibfield  {title}
  {\bibinfo {title} {A 98-qubit trapped-ion quantum computer with all-to-all
  connectivity},\ }\href {https://doi.org/10.1038/s41586-026-10676-4}
  {\bibfield  {journal} {\bibinfo  {journal} {Nature}\ }\textbf {\bibinfo
  {volume} {655}},\ \bibinfo {pages} {81} (\bibinfo {year} {2026})}\BibitemShut
  {NoStop}%
\bibitem [{\citenamefont {Beugnon}\ \emph {et~al.}(2007)\citenamefont
  {Beugnon}, \citenamefont {Tuchendler}, \citenamefont {Marion}, \citenamefont
  {Ga{\"e}tan}, \citenamefont {Miroshnychenko}, \citenamefont {Sortais},
  \citenamefont {Lance}, \citenamefont {Jones}, \citenamefont {Messin},
  \citenamefont {Browaeys},\ and\ \citenamefont
  {Grangier}}]{Beugnon2007TweezerTransport}%
  \BibitemOpen
  \bibfield  {author} {\bibinfo {author} {\bibfnamefont {J.}~\bibnamefont
  {Beugnon}}, \bibinfo {author} {\bibfnamefont {C.}~\bibnamefont {Tuchendler}},
  \bibinfo {author} {\bibfnamefont {H.}~\bibnamefont {Marion}}, \bibinfo
  {author} {\bibfnamefont {A.}~\bibnamefont {Ga{\"e}tan}}, \bibinfo {author}
  {\bibfnamefont {Y.}~\bibnamefont {Miroshnychenko}}, \bibinfo {author}
  {\bibfnamefont {Y.~R.~P.}\ \bibnamefont {Sortais}}, \bibinfo {author}
  {\bibfnamefont {A.~M.}\ \bibnamefont {Lance}}, \bibinfo {author}
  {\bibfnamefont {M.~P.~A.}\ \bibnamefont {Jones}}, \bibinfo {author}
  {\bibfnamefont {G.}~\bibnamefont {Messin}}, \bibinfo {author} {\bibfnamefont
  {A.}~\bibnamefont {Browaeys}},\ and\ \bibinfo {author} {\bibfnamefont
  {P.}~\bibnamefont {Grangier}},\ }\bibfield  {title} {\bibinfo {title}
  {Two-dimensional transport and transfer of a single atomic qubit in optical
  tweezers},\ }\href {https://doi.org/10.1038/nphys698} {\bibfield  {journal}
  {\bibinfo  {journal} {Nature Physics}\ }\textbf {\bibinfo {volume} {3}},\
  \bibinfo {pages} {696} (\bibinfo {year} {2007})}\BibitemShut {NoStop}%
\bibitem [{\citenamefont {Barredo}\ \emph {et~al.}(2016)\citenamefont
  {Barredo}, \citenamefont {de~L{\'e}s{\'e}leuc}, \citenamefont {Lienhard},
  \citenamefont {Lahaye},\ and\ \citenamefont
  {Browaeys}}]{Barredo2016Assembler}%
  \BibitemOpen
  \bibfield  {author} {\bibinfo {author} {\bibfnamefont {D.}~\bibnamefont
  {Barredo}}, \bibinfo {author} {\bibfnamefont {S.}~\bibnamefont
  {de~L{\'e}s{\'e}leuc}}, \bibinfo {author} {\bibfnamefont {V.}~\bibnamefont
  {Lienhard}}, \bibinfo {author} {\bibfnamefont {T.}~\bibnamefont {Lahaye}},\
  and\ \bibinfo {author} {\bibfnamefont {A.}~\bibnamefont {Browaeys}},\
  }\bibfield  {title} {\bibinfo {title} {An atom-by-atom assembler of
  defect-free arbitrary two-dimensional atomic arrays},\ }\href
  {https://doi.org/10.1126/science.aah3778} {\bibfield  {journal} {\bibinfo
  {journal} {Science}\ }\textbf {\bibinfo {volume} {354}},\ \bibinfo {pages}
  {1021} (\bibinfo {year} {2016})}\BibitemShut {NoStop}%
\bibitem [{\citenamefont {Endres}\ \emph {et~al.}(2016)\citenamefont {Endres},
  \citenamefont {Bernien}, \citenamefont {Keesling}, \citenamefont {Levine},
  \citenamefont {Anschuetz}, \citenamefont {Krajenbrink}, \citenamefont
  {Senko}, \citenamefont {Vuleti{\'c}}, \citenamefont {Greiner},\ and\
  \citenamefont {Lukin}}]{Endres2016Assembly}%
  \BibitemOpen
  \bibfield  {author} {\bibinfo {author} {\bibfnamefont {M.}~\bibnamefont
  {Endres}}, \bibinfo {author} {\bibfnamefont {H.}~\bibnamefont {Bernien}},
  \bibinfo {author} {\bibfnamefont {A.}~\bibnamefont {Keesling}}, \bibinfo
  {author} {\bibfnamefont {H.}~\bibnamefont {Levine}}, \bibinfo {author}
  {\bibfnamefont {E.~R.}\ \bibnamefont {Anschuetz}}, \bibinfo {author}
  {\bibfnamefont {A.}~\bibnamefont {Krajenbrink}}, \bibinfo {author}
  {\bibfnamefont {C.}~\bibnamefont {Senko}}, \bibinfo {author} {\bibfnamefont
  {V.}~\bibnamefont {Vuleti{\'c}}}, \bibinfo {author} {\bibfnamefont
  {M.}~\bibnamefont {Greiner}},\ and\ \bibinfo {author} {\bibfnamefont {M.~D.}\
  \bibnamefont {Lukin}},\ }\bibfield  {title} {\bibinfo {title} {Atom-by-atom
  assembly of defect-free one-dimensional cold atom arrays},\ }\href
  {https://doi.org/10.1126/science.aah3752} {\bibfield  {journal} {\bibinfo
  {journal} {Science}\ }\textbf {\bibinfo {volume} {354}},\ \bibinfo {pages}
  {1024} (\bibinfo {year} {2016})}\BibitemShut {NoStop}%
\bibitem [{\citenamefont {Browaeys}\ and\ \citenamefont
  {Lahaye}(2020)}]{BrowaeysLahaye2020RydbergReview}%
  \BibitemOpen
  \bibfield  {author} {\bibinfo {author} {\bibfnamefont {A.}~\bibnamefont
  {Browaeys}}\ and\ \bibinfo {author} {\bibfnamefont {T.}~\bibnamefont
  {Lahaye}},\ }\bibfield  {title} {\bibinfo {title} {Many-body physics with
  individually controlled rydberg atoms},\ }\href
  {https://doi.org/10.1038/s41567-019-0733-z} {\bibfield  {journal} {\bibinfo
  {journal} {Nature Physics}\ }\textbf {\bibinfo {volume} {16}},\ \bibinfo
  {pages} {132} (\bibinfo {year} {2020})}\BibitemShut {NoStop}%
\bibitem [{\citenamefont {Bluvstein}\ \emph {et~al.}(2022)\citenamefont
  {Bluvstein}, \citenamefont {Levine}, \citenamefont {Semeghini}, \citenamefont
  {Wang}, \citenamefont {Ebadi}, \citenamefont {Kalinowski}, \citenamefont
  {Keesling}, \citenamefont {Maskara}, \citenamefont {Pichler}, \citenamefont
  {Greiner}, \citenamefont {Vuleti{\'c}},\ and\ \citenamefont
  {Lukin}}]{bluvstein2022quantum}%
  \BibitemOpen
  \bibfield  {author} {\bibinfo {author} {\bibfnamefont {D.}~\bibnamefont
  {Bluvstein}}, \bibinfo {author} {\bibfnamefont {H.}~\bibnamefont {Levine}},
  \bibinfo {author} {\bibfnamefont {G.}~\bibnamefont {Semeghini}}, \bibinfo
  {author} {\bibfnamefont {T.~T.}\ \bibnamefont {Wang}}, \bibinfo {author}
  {\bibfnamefont {S.}~\bibnamefont {Ebadi}}, \bibinfo {author} {\bibfnamefont
  {M.}~\bibnamefont {Kalinowski}}, \bibinfo {author} {\bibfnamefont
  {A.}~\bibnamefont {Keesling}}, \bibinfo {author} {\bibfnamefont
  {N.}~\bibnamefont {Maskara}}, \bibinfo {author} {\bibfnamefont
  {H.}~\bibnamefont {Pichler}}, \bibinfo {author} {\bibfnamefont
  {M.}~\bibnamefont {Greiner}}, \bibinfo {author} {\bibfnamefont
  {V.}~\bibnamefont {Vuleti{\'c}}},\ and\ \bibinfo {author} {\bibfnamefont
  {M.~D.}\ \bibnamefont {Lukin}},\ }\bibfield  {title} {\bibinfo {title} {A
  quantum processor based on coherent transport of entangled atom arrays},\
  }\href {https://doi.org/10.1038/s41586-022-04592-6} {\bibfield  {journal}
  {\bibinfo  {journal} {Nature}\ }\textbf {\bibinfo {volume} {604}},\ \bibinfo
  {pages} {451} (\bibinfo {year} {2022})}\BibitemShut {NoStop}%
\bibitem [{\citenamefont {Sales~Rodriguez}\ \emph {et~al.}(2025)\citenamefont
  {Sales~Rodriguez}, \citenamefont {Robinson}, \citenamefont {Jepsen} \emph
  {et~al.}}]{salesrodriguez2025magic}%
  \BibitemOpen
  \bibfield  {author} {\bibinfo {author} {\bibfnamefont {P.}~\bibnamefont
  {Sales~Rodriguez}}, \bibinfo {author} {\bibfnamefont {J.~M.}\ \bibnamefont
  {Robinson}}, \bibinfo {author} {\bibfnamefont {P.~N.}\ \bibnamefont
  {Jepsen}}, \emph {et~al.},\ }\bibfield  {title} {\bibinfo {title}
  {Experimental demonstration of logical magic state distillation},\ }\href
  {https://doi.org/10.1038/s41586-025-09367-3} {\bibfield  {journal} {\bibinfo
  {journal} {Nature}\ }\textbf {\bibinfo {volume} {645}},\ \bibinfo {pages}
  {620} (\bibinfo {year} {2025})}\BibitemShut {NoStop}%
\bibitem [{\citenamefont {Constantinides}\ \emph {et~al.}(2024)\citenamefont
  {Constantinides}, \citenamefont {Fahimniya}, \citenamefont {Devulapalli},
  \citenamefont {Bluvstein}, \citenamefont {Gullans}, \citenamefont {Porto},
  \citenamefont {Childs},\ and\ \citenamefont
  {Gorshkov}}]{constantinides2024optimal}%
  \BibitemOpen
  \bibfield  {author} {\bibinfo {author} {\bibfnamefont {N.}~\bibnamefont
  {Constantinides}}, \bibinfo {author} {\bibfnamefont {A.}~\bibnamefont
  {Fahimniya}}, \bibinfo {author} {\bibfnamefont {D.}~\bibnamefont
  {Devulapalli}}, \bibinfo {author} {\bibfnamefont {D.}~\bibnamefont
  {Bluvstein}}, \bibinfo {author} {\bibfnamefont {M.~J.}\ \bibnamefont
  {Gullans}}, \bibinfo {author} {\bibfnamefont {J.}~\bibnamefont {Porto}},
  \bibinfo {author} {\bibfnamefont {A.~M.}\ \bibnamefont {Childs}},\ and\
  \bibinfo {author} {\bibfnamefont {A.~V.}\ \bibnamefont {Gorshkov}},\
  }\bibfield  {title} {\bibinfo {title} {Optimal routing protocols for
  reconfigurable atom arrays},\ }\href {https://arxiv.org/abs/2411.05061}
  {\bibfield  {journal} {\bibinfo  {journal} {arXiv preprint arXiv:2411.05061}\
  } (\bibinfo {year} {2024})}\BibitemShut {NoStop}%
\bibitem [{\citenamefont {Monroe}\ \emph {et~al.}(1996)\citenamefont {Monroe},
  \citenamefont {Meekhof}, \citenamefont {King},\ and\ \citenamefont
  {Wineland}}]{Monroe1996CatState}%
  \BibitemOpen
  \bibfield  {author} {\bibinfo {author} {\bibfnamefont {C.}~\bibnamefont
  {Monroe}}, \bibinfo {author} {\bibfnamefont {D.~M.}\ \bibnamefont {Meekhof}},
  \bibinfo {author} {\bibfnamefont {B.~E.}\ \bibnamefont {King}},\ and\
  \bibinfo {author} {\bibfnamefont {D.~J.}\ \bibnamefont {Wineland}},\
  }\bibfield  {title} {\bibinfo {title} {A {``Schr{\"o}dinger Cat''}
  superposition state of an atom},\ }\href
  {https://doi.org/10.1126/science.272.5265.1131} {\bibfield  {journal}
  {\bibinfo  {journal} {Science}\ }\textbf {\bibinfo {volume} {272}},\ \bibinfo
  {pages} {1131} (\bibinfo {year} {1996})}\BibitemShut {NoStop}%
\bibitem [{\citenamefont {Saner}\ \emph {et~al.}(2026)\citenamefont {Saner},
  \citenamefont {B{\u{a}}z{\u{a}}van}, \citenamefont {Webb}, \citenamefont
  {Araneda}, \citenamefont {Lucas}, \citenamefont {Ballance},\ and\
  \citenamefont {Srinivas}}]{saner2026generating}%
  \BibitemOpen
  \bibfield  {author} {\bibinfo {author} {\bibfnamefont {S.}~\bibnamefont
  {Saner}}, \bibinfo {author} {\bibfnamefont {O.}~\bibnamefont
  {B{\u{a}}z{\u{a}}van}}, \bibinfo {author} {\bibfnamefont {D.}~\bibnamefont
  {Webb}}, \bibinfo {author} {\bibfnamefont {G.}~\bibnamefont {Araneda}},
  \bibinfo {author} {\bibfnamefont {D.}~\bibnamefont {Lucas}}, \bibinfo
  {author} {\bibfnamefont {C.}~\bibnamefont {Ballance}},\ and\ \bibinfo
  {author} {\bibfnamefont {R.}~\bibnamefont {Srinivas}},\ }\bibfield  {title}
  {\bibinfo {title} {Generating arbitrary superpositions of nonclassical
  quantum harmonic oscillator states},\ }\href
  {https://doi.org/10.1103/k1xk-yt42} {\bibfield  {journal} {\bibinfo
  {journal} {Physical Review X}\ }\textbf {\bibinfo {volume} {16}},\ \bibinfo
  {pages} {021049} (\bibinfo {year} {2026})}\BibitemShut {NoStop}%
\bibitem [{\citenamefont {Leindecker}\ \emph {et~al.}(2026)\citenamefont
  {Leindecker}, \citenamefont {Milanovic}, \citenamefont {Behrle},
  \citenamefont {Brucke}, \citenamefont {Marinelli}, \citenamefont {Schmidt},
  \citenamefont {Home},\ and\ \citenamefont {Hempel}}]{leindecker2026state}%
  \BibitemOpen
  \bibfield  {author} {\bibinfo {author} {\bibfnamefont {P.}~\bibnamefont
  {Leindecker}}, \bibinfo {author} {\bibfnamefont {L.}~\bibnamefont
  {Milanovic}}, \bibinfo {author} {\bibfnamefont {T.}~\bibnamefont {Behrle}},
  \bibinfo {author} {\bibfnamefont {E.}~\bibnamefont {Brucke}}, \bibinfo
  {author} {\bibfnamefont {M.}~\bibnamefont {Marinelli}}, \bibinfo {author}
  {\bibfnamefont {J.}~\bibnamefont {Schmidt}}, \bibinfo {author} {\bibfnamefont
  {J.}~\bibnamefont {Home}},\ and\ \bibinfo {author} {\bibfnamefont
  {C.}~\bibnamefont {Hempel}},\ }\bibfield  {title} {\bibinfo {title}
  {State-dependent gaussian gate set using an optical tweezer for trapped
  ions},\ }\href {https://arxiv.org/abs/2606.31864} {\bibfield  {journal}
  {\bibinfo  {journal} {arXiv preprint arXiv:2606.31864}\ } (\bibinfo {year}
  {2026})}\BibitemShut {NoStop}%
\bibitem [{\citenamefont {Diaz}\ \emph {et~al.}(2026)\citenamefont {Diaz},
  \citenamefont {Mohammadi-Arzanagh}, \citenamefont {Zhu}, \citenamefont
  {Hafezi}, \citenamefont {Linke}, \citenamefont {Green},\ and\ \citenamefont
  {Nam}}]{diaz2026arbitrary}%
  \BibitemOpen
  \bibfield  {author} {\bibinfo {author} {\bibfnamefont {M.}~\bibnamefont
  {Diaz}}, \bibinfo {author} {\bibfnamefont {M.}~\bibnamefont
  {Mohammadi-Arzanagh}}, \bibinfo {author} {\bibfnamefont {Y.}~\bibnamefont
  {Zhu}}, \bibinfo {author} {\bibfnamefont {M.}~\bibnamefont {Hafezi}},
  \bibinfo {author} {\bibfnamefont {N.~M.}\ \bibnamefont {Linke}}, \bibinfo
  {author} {\bibfnamefont {A.~M.}\ \bibnamefont {Green}},\ and\ \bibinfo
  {author} {\bibfnamefont {A.~Y.}\ \bibnamefont {Nam}},\ }\bibfield  {title}
  {\bibinfo {title} {Arbitrary parallel entangling gates with independent
  calibration on a trapped ion quantum computer},\ }\href
  {https://arxiv.org/abs/2604.25993} {\bibfield  {journal} {\bibinfo  {journal}
  {arXiv preprint arXiv:2604.25993}\ } (\bibinfo {year} {2026})}\BibitemShut
  {NoStop}%
\bibitem [{\citenamefont {Mandel}\ \emph {et~al.}(2003)\citenamefont {Mandel},
  \citenamefont {Greiner}, \citenamefont {Widera}, \citenamefont {Rom},
  \citenamefont {H{\"a}nsch},\ and\ \citenamefont
  {Bloch}}]{Mandel2003SpinDependentTransport}%
  \BibitemOpen
  \bibfield  {author} {\bibinfo {author} {\bibfnamefont {O.}~\bibnamefont
  {Mandel}}, \bibinfo {author} {\bibfnamefont {M.}~\bibnamefont {Greiner}},
  \bibinfo {author} {\bibfnamefont {A.}~\bibnamefont {Widera}}, \bibinfo
  {author} {\bibfnamefont {T.}~\bibnamefont {Rom}}, \bibinfo {author}
  {\bibfnamefont {T.~W.}\ \bibnamefont {H{\"a}nsch}},\ and\ \bibinfo {author}
  {\bibfnamefont {I.}~\bibnamefont {Bloch}},\ }\bibfield  {title} {\bibinfo
  {title} {Coherent transport of neutral atoms in spin-dependent optical
  lattice potentials},\ }\href {https://doi.org/10.1103/PhysRevLett.91.010407}
  {\bibfield  {journal} {\bibinfo  {journal} {Physical Review Letters}\
  }\textbf {\bibinfo {volume} {91}},\ \bibinfo {pages} {010407} (\bibinfo
  {year} {2003})}\BibitemShut {NoStop}%
\bibitem [{\citenamefont {Greve}\ \emph {et~al.}(2022)\citenamefont {Greve},
  \citenamefont {Luo}, \citenamefont {Wu},\ and\ \citenamefont
  {Thompson}}]{greve_entanglement-enhanced_2022-1}%
  \BibitemOpen
  \bibfield  {author} {\bibinfo {author} {\bibfnamefont {G.~P.}\ \bibnamefont
  {Greve}}, \bibinfo {author} {\bibfnamefont {C.}~\bibnamefont {Luo}}, \bibinfo
  {author} {\bibfnamefont {B.}~\bibnamefont {Wu}},\ and\ \bibinfo {author}
  {\bibfnamefont {J.~K.}\ \bibnamefont {Thompson}},\ }\bibfield  {title}
  {\bibinfo {title} {Entanglement-enhanced matter-wave interferometry in a
  high-finesse cavity},\ }\href {https://doi.org/10.1038/s41586-022-05197-9}
  {\bibfield  {journal} {\bibinfo  {journal} {Nature}\ }\textbf {\bibinfo
  {volume} {610}},\ \bibinfo {pages} {472} (\bibinfo {year}
  {2022})}\BibitemShut {NoStop}%
\bibitem [{\citenamefont {Aeppli}\ \emph {et~al.}(2022)\citenamefont {Aeppli},
  \citenamefont {Chu}, \citenamefont {Bothwell}, \citenamefont {Kennedy},
  \citenamefont {Kedar}, \citenamefont {He}, \citenamefont {Rey},\ and\
  \citenamefont {Ye}}]{aeppli_hamiltonian_2022}%
  \BibitemOpen
  \bibfield  {author} {\bibinfo {author} {\bibfnamefont {A.}~\bibnamefont
  {Aeppli}}, \bibinfo {author} {\bibfnamefont {A.}~\bibnamefont {Chu}},
  \bibinfo {author} {\bibfnamefont {T.}~\bibnamefont {Bothwell}}, \bibinfo
  {author} {\bibfnamefont {C.~J.}\ \bibnamefont {Kennedy}}, \bibinfo {author}
  {\bibfnamefont {D.}~\bibnamefont {Kedar}}, \bibinfo {author} {\bibfnamefont
  {P.}~\bibnamefont {He}}, \bibinfo {author} {\bibfnamefont {A.~M.}\
  \bibnamefont {Rey}},\ and\ \bibinfo {author} {\bibfnamefont {J.}~\bibnamefont
  {Ye}},\ }\bibfield  {title} {\bibinfo {title} {Hamiltonian engineering of
  spin-orbit--coupled fermions in a {{Wannier-Stark}} optical lattice clock},\
  }\href {https://doi.org/10.1126/sciadv.adc9242} {\bibfield  {journal}
  {\bibinfo  {journal} {Science Advances}\ }\textbf {\bibinfo {volume} {8}},\
  \bibinfo {pages} {eadc9242} (\bibinfo {year} {2022})}\BibitemShut {NoStop}%
\bibitem [{\citenamefont {Malia}\ \emph {et~al.}(2022)\citenamefont {Malia},
  \citenamefont {Wu}, \citenamefont {{Mart{\'i}nez-Rinc{\'o}n}},\ and\
  \citenamefont {Kasevich}}]{malia_distributed_2022}%
  \BibitemOpen
  \bibfield  {author} {\bibinfo {author} {\bibfnamefont {B.~K.}\ \bibnamefont
  {Malia}}, \bibinfo {author} {\bibfnamefont {Y.}~\bibnamefont {Wu}}, \bibinfo
  {author} {\bibfnamefont {J.}~\bibnamefont {{Mart{\'i}nez-Rinc{\'o}n}}},\ and\
  \bibinfo {author} {\bibfnamefont {M.~A.}\ \bibnamefont {Kasevich}},\
  }\bibfield  {title} {\bibinfo {title} {Distributed quantum sensing with
  mode-entangled spin-squeezed atomic states},\ }\href
  {https://doi.org/10.1038/s41586-022-05363-z} {\bibfield  {journal} {\bibinfo
  {journal} {Nature}\ }\textbf {\bibinfo {volume} {612}},\ \bibinfo {pages}
  {661} (\bibinfo {year} {2022})}\BibitemShut {NoStop}%
\bibitem [{\citenamefont {Bluvstein}\ \emph {et~al.}(2026)\citenamefont
  {Bluvstein}, \citenamefont {Geim}, \citenamefont {Li}, \citenamefont
  {Evered}, \citenamefont {Bonilla~Ataides}, \citenamefont {Baranes},
  \citenamefont {Gu}, \citenamefont {Manovitz}, \citenamefont {Xu},
  \citenamefont {Kalinowski} \emph {et~al.}}]{bluvstein2026fault}%
  \BibitemOpen
  \bibfield  {author} {\bibinfo {author} {\bibfnamefont {D.}~\bibnamefont
  {Bluvstein}}, \bibinfo {author} {\bibfnamefont {A.~A.}\ \bibnamefont {Geim}},
  \bibinfo {author} {\bibfnamefont {S.~H.}\ \bibnamefont {Li}}, \bibinfo
  {author} {\bibfnamefont {S.~J.}\ \bibnamefont {Evered}}, \bibinfo {author}
  {\bibfnamefont {J.~P.}\ \bibnamefont {Bonilla~Ataides}}, \bibinfo {author}
  {\bibfnamefont {G.}~\bibnamefont {Baranes}}, \bibinfo {author} {\bibfnamefont
  {A.}~\bibnamefont {Gu}}, \bibinfo {author} {\bibfnamefont {T.}~\bibnamefont
  {Manovitz}}, \bibinfo {author} {\bibfnamefont {M.}~\bibnamefont {Xu}},
  \bibinfo {author} {\bibfnamefont {M.}~\bibnamefont {Kalinowski}}, \emph
  {et~al.},\ }\bibfield  {title} {\bibinfo {title} {A fault-tolerant
  neutral-atom architecture for universal quantum computation},\ }\href
  {https://doi.org/10.1038/s41586-025-09848-5} {\bibfield  {journal} {\bibinfo
  {journal} {Nature}\ }\textbf {\bibinfo {volume} {649}},\ \bibinfo {pages}
  {39} (\bibinfo {year} {2026})}\BibitemShut {NoStop}%
\bibitem [{\citenamefont {Pampel}\ \emph {et~al.}(2026)\citenamefont {Pampel},
  \citenamefont {Lubin}, \citenamefont {Hewatt}, \citenamefont {McCabe},
  \citenamefont {Hwang}, \citenamefont {Muleady}, \citenamefont {Xu},
  \citenamefont {Rey},\ and\ \citenamefont {Regal}}]{pampel2026motional}%
  \BibitemOpen
  \bibfield  {author} {\bibinfo {author} {\bibfnamefont {S.~K.}\ \bibnamefont
  {Pampel}}, \bibinfo {author} {\bibfnamefont {G.}~\bibnamefont {Lubin}},
  \bibinfo {author} {\bibfnamefont {D.~P.}\ \bibnamefont {Hewatt}}, \bibinfo
  {author} {\bibfnamefont {C.}~\bibnamefont {McCabe}}, \bibinfo {author}
  {\bibfnamefont {J.}~\bibnamefont {Hwang}}, \bibinfo {author} {\bibfnamefont
  {S.~R.}\ \bibnamefont {Muleady}}, \bibinfo {author} {\bibfnamefont
  {T.}~\bibnamefont {Xu}}, \bibinfo {author} {\bibfnamefont {A.~M.}\
  \bibnamefont {Rey}},\ and\ \bibinfo {author} {\bibfnamefont {C.~A.}\
  \bibnamefont {Regal}},\ }\bibfield  {title} {\bibinfo {title} {Motional
  {K}err-cat states of an atom in an optical tweezer},\ }\href
  {https://arxiv.org/abs/2607.18579} {\bibfield  {journal} {\bibinfo  {journal}
  {arXiv preprint arXiv:2607.18579}\ } (\bibinfo {year} {2026})}\BibitemShut
  {NoStop}%
\bibitem [{\citenamefont {Manetsch}\ \emph {et~al.}(2025)\citenamefont
  {Manetsch}, \citenamefont {Nomura}, \citenamefont {Bataille}, \citenamefont
  {Lv}, \citenamefont {Leung},\ and\ \citenamefont
  {Endres}}]{manetsch2025tweezer}%
  \BibitemOpen
  \bibfield  {author} {\bibinfo {author} {\bibfnamefont {H.~J.}\ \bibnamefont
  {Manetsch}}, \bibinfo {author} {\bibfnamefont {G.}~\bibnamefont {Nomura}},
  \bibinfo {author} {\bibfnamefont {E.}~\bibnamefont {Bataille}}, \bibinfo
  {author} {\bibfnamefont {X.}~\bibnamefont {Lv}}, \bibinfo {author}
  {\bibfnamefont {K.~H.}\ \bibnamefont {Leung}},\ and\ \bibinfo {author}
  {\bibfnamefont {M.}~\bibnamefont {Endres}},\ }\bibfield  {title} {\bibinfo
  {title} {A tweezer array with 6,100 highly coherent atomic qubits},\ }\href
  {https://doi.org/10.1038/s41586-025-09641-4} {\bibfield  {journal} {\bibinfo
  {journal} {Nature}\ }\textbf {\bibinfo {volume} {647}},\ \bibinfo {pages}
  {60} (\bibinfo {year} {2025})}\BibitemShut {NoStop}%
\bibitem [{\citenamefont {Karski}\ \emph {et~al.}(2009)\citenamefont {Karski},
  \citenamefont {F{\"o}rster}, \citenamefont {Choi}, \citenamefont {Steffen},
  \citenamefont {Alt}, \citenamefont {Meschede},\ and\ \citenamefont
  {Widera}}]{Karski2009QuantumWalk}%
  \BibitemOpen
  \bibfield  {author} {\bibinfo {author} {\bibfnamefont {M.}~\bibnamefont
  {Karski}}, \bibinfo {author} {\bibfnamefont {L.}~\bibnamefont {F{\"o}rster}},
  \bibinfo {author} {\bibfnamefont {J.-M.}\ \bibnamefont {Choi}}, \bibinfo
  {author} {\bibfnamefont {A.}~\bibnamefont {Steffen}}, \bibinfo {author}
  {\bibfnamefont {W.}~\bibnamefont {Alt}}, \bibinfo {author} {\bibfnamefont
  {D.}~\bibnamefont {Meschede}},\ and\ \bibinfo {author} {\bibfnamefont
  {A.}~\bibnamefont {Widera}},\ }\bibfield  {title} {\bibinfo {title} {Quantum
  walk in position space with single optically trapped atoms},\ }\href
  {https://doi.org/10.1126/science.1174436} {\bibfield  {journal} {\bibinfo
  {journal} {Science}\ }\textbf {\bibinfo {volume} {325}},\ \bibinfo {pages}
  {174} (\bibinfo {year} {2009})}\BibitemShut {NoStop}%
\bibitem [{\citenamefont {Schmitz}\ \emph {et~al.}(2009)\citenamefont
  {Schmitz}, \citenamefont {Matjeschk}, \citenamefont {Schneider},
  \citenamefont {Glueckert}, \citenamefont {Enderlein}, \citenamefont {Huber},\
  and\ \citenamefont {Schaetz}}]{Schmitz2009IonQuantumWalk}%
  \BibitemOpen
  \bibfield  {author} {\bibinfo {author} {\bibfnamefont {H.}~\bibnamefont
  {Schmitz}}, \bibinfo {author} {\bibfnamefont {R.}~\bibnamefont {Matjeschk}},
  \bibinfo {author} {\bibfnamefont {C.}~\bibnamefont {Schneider}}, \bibinfo
  {author} {\bibfnamefont {J.}~\bibnamefont {Glueckert}}, \bibinfo {author}
  {\bibfnamefont {M.}~\bibnamefont {Enderlein}}, \bibinfo {author}
  {\bibfnamefont {T.}~\bibnamefont {Huber}},\ and\ \bibinfo {author}
  {\bibfnamefont {T.}~\bibnamefont {Schaetz}},\ }\bibfield  {title} {\bibinfo
  {title} {Quantum walk of a trapped ion in phase space},\ }\href
  {https://doi.org/10.1103/PhysRevLett.103.090504} {\bibfield  {journal}
  {\bibinfo  {journal} {Physical Review Letters}\ }\textbf {\bibinfo {volume}
  {103}},\ \bibinfo {pages} {090504} (\bibinfo {year} {2009})}\BibitemShut
  {NoStop}%
\bibitem [{\citenamefont {Young}\ \emph {et~al.}(2022)\citenamefont {Young},
  \citenamefont {Eckner}, \citenamefont {Schine}, \citenamefont {Childs},\ and\
  \citenamefont {Kaufman}}]{young_tweezer-programmable_2022-1}%
  \BibitemOpen
  \bibfield  {author} {\bibinfo {author} {\bibfnamefont {A.~W.}\ \bibnamefont
  {Young}}, \bibinfo {author} {\bibfnamefont {W.~J.}\ \bibnamefont {Eckner}},
  \bibinfo {author} {\bibfnamefont {N.}~\bibnamefont {Schine}}, \bibinfo
  {author} {\bibfnamefont {A.~M.}\ \bibnamefont {Childs}},\ and\ \bibinfo
  {author} {\bibfnamefont {A.~M.}\ \bibnamefont {Kaufman}},\ }\bibfield
  {title} {\bibinfo {title} {Tweezer-programmable {{2D}} quantum walks in a
  {{Hubbard-regime}} lattice},\ }\href
  {https://doi.org/10.1126/science.abo0608} {\bibfield  {journal} {\bibinfo
  {journal} {Science}\ }\textbf {\bibinfo {volume} {377}},\ \bibinfo {pages}
  {885} (\bibinfo {year} {2022})}\BibitemShut {NoStop}%
\bibitem [{\citenamefont {Grover}(1996)}]{Grover1996STOC}%
  \BibitemOpen
  \bibfield  {author} {\bibinfo {author} {\bibfnamefont {L.~K.}\ \bibnamefont
  {Grover}},\ }\bibfield  {title} {\bibinfo {title} {A fast quantum mechanical
  algorithm for database search},\ }in\ \href
  {https://doi.org/10.1145/237814.237866} {\emph {\bibinfo {booktitle}
  {Proceedings of the 28th Annual ACM Symposium on Theory of Computing
  (STOC)}}}\ (\bibinfo  {publisher} {ACM},\ \bibinfo {year} {1996})\ pp.\
  \bibinfo {pages} {212--219}\BibitemShut {NoStop}%
\bibitem [{\citenamefont {Boyer}\ \emph {et~al.}(1998)\citenamefont {Boyer},
  \citenamefont {Brassard}, \citenamefont {H{\o}yer},\ and\ \citenamefont
  {Tapp}}]{Boyer1998TightBounds}%
  \BibitemOpen
  \bibfield  {author} {\bibinfo {author} {\bibfnamefont {M.}~\bibnamefont
  {Boyer}}, \bibinfo {author} {\bibfnamefont {G.}~\bibnamefont {Brassard}},
  \bibinfo {author} {\bibfnamefont {P.}~\bibnamefont {H{\o}yer}},\ and\
  \bibinfo {author} {\bibfnamefont {A.}~\bibnamefont {Tapp}},\ }\bibfield
  {title} {\bibinfo {title} {Tight bounds on quantum searching},\ }\href
  {https://doi.org/10.1002/(SICI)1521-3978(199806)46:4/5<493::AID-PROP493>3.0.CO;2-P}
  {\bibfield  {journal} {\bibinfo  {journal} {Fortschritte der Physik}\
  }\textbf {\bibinfo {volume} {46}},\ \bibinfo {pages} {493} (\bibinfo {year}
  {1998})}\BibitemShut {NoStop}%
\bibitem [{\citenamefont {Zalka}(1999)}]{Zalka1999Optimal}%
  \BibitemOpen
  \bibfield  {author} {\bibinfo {author} {\bibfnamefont {C.}~\bibnamefont
  {Zalka}},\ }\bibfield  {title} {\bibinfo {title} {Grover's quantum searching
  algorithm is optimal},\ }\href {https://doi.org/10.1103/PhysRevA.60.2746}
  {\bibfield  {journal} {\bibinfo  {journal} {Physical Review A}\ }\textbf
  {\bibinfo {volume} {60}},\ \bibinfo {pages} {2746} (\bibinfo {year}
  {1999})}\BibitemShut {NoStop}%
\bibitem [{\citenamefont {Farhi}\ and\ \citenamefont
  {Gutmann}(1998)}]{farhi1998analog}%
  \BibitemOpen
  \bibfield  {author} {\bibinfo {author} {\bibfnamefont {E.}~\bibnamefont
  {Farhi}}\ and\ \bibinfo {author} {\bibfnamefont {S.}~\bibnamefont
  {Gutmann}},\ }\bibfield  {title} {\bibinfo {title} {Analog analogue of a
  digital quantum computation},\ }\href
  {https://doi.org/10.1103/PhysRevA.57.2403} {\bibfield  {journal} {\bibinfo
  {journal} {Physical Review A}\ }\textbf {\bibinfo {volume} {57}},\ \bibinfo
  {pages} {2403} (\bibinfo {year} {1998})}\BibitemShut {NoStop}%
\bibitem [{\citenamefont {Childs}\ \emph {et~al.}(2000)\citenamefont {Childs},
  \citenamefont {Preskill},\ and\ \citenamefont
  {Renes}}]{Childs2000PrecisionMeasurement}%
  \BibitemOpen
  \bibfield  {author} {\bibinfo {author} {\bibfnamefont {A.~M.}\ \bibnamefont
  {Childs}}, \bibinfo {author} {\bibfnamefont {J.}~\bibnamefont {Preskill}},\
  and\ \bibinfo {author} {\bibfnamefont {J.}~\bibnamefont {Renes}},\ }\bibfield
   {title} {\bibinfo {title} {Quantum information and precision measurement},\
  }\href {https://doi.org/10.1080/095003400148123} {\bibfield  {journal}
  {\bibinfo  {journal} {Journal of Modern Optics}\ }\textbf {\bibinfo {volume}
  {47}},\ \bibinfo {pages} {155} (\bibinfo {year} {2000})}\BibitemShut
  {NoStop}%
\bibitem [{\citenamefont {Benioff}(2002)}]{benioff2002space}%
  \BibitemOpen
  \bibfield  {author} {\bibinfo {author} {\bibfnamefont {P.}~\bibnamefont
  {Benioff}},\ }\bibfield  {title} {\bibinfo {title} {Space searches with a
  quantum robot},\ }in\ \href {https://arxiv.org/abs/quant-ph/0003006} {\emph
  {\bibinfo {booktitle} {Quantum Computation and Information}}},\ \bibinfo
  {series} {Contemporary Mathematics}, Vol.\ \bibinfo {volume} {305}\ (\bibinfo
   {publisher} {American Mathematical Society},\ \bibinfo {year}
  {2002})\BibitemShut {NoStop}%
\bibitem [{\citenamefont {Aaronson}\ and\ \citenamefont
  {Ambainis}(2005)}]{aaronson2005quantum}%
  \BibitemOpen
  \bibfield  {author} {\bibinfo {author} {\bibfnamefont {S.}~\bibnamefont
  {Aaronson}}\ and\ \bibinfo {author} {\bibfnamefont {A.}~\bibnamefont
  {Ambainis}},\ }\bibfield  {title} {\bibinfo {title} {Quantum search of
  spatial regions},\ }\href {https://doi.org/10.4086/toc.2005.v001a004}
  {\bibfield  {journal} {\bibinfo  {journal} {Theory of Computing}\ }\textbf
  {\bibinfo {volume} {1}},\ \bibinfo {pages} {47} (\bibinfo {year}
  {2005})}\BibitemShut {NoStop}%
\bibitem [{\citenamefont {Childs}\ and\ \citenamefont
  {Goldstone}(2004)}]{childs2004spatial}%
  \BibitemOpen
  \bibfield  {author} {\bibinfo {author} {\bibfnamefont {A.~M.}\ \bibnamefont
  {Childs}}\ and\ \bibinfo {author} {\bibfnamefont {J.}~\bibnamefont
  {Goldstone}},\ }\bibfield  {title} {\bibinfo {title} {Spatial search by
  quantum walk},\ }\href {https://doi.org/10.1103/PhysRevA.70.022314}
  {\bibfield  {journal} {\bibinfo  {journal} {Physical Review A}\ }\textbf
  {\bibinfo {volume} {70}},\ \bibinfo {pages} {022314} (\bibinfo {year}
  {2004})}\BibitemShut {NoStop}%
\bibitem [{\citenamefont {Allen}\ \emph {et~al.}(2025)\citenamefont {Allen},
  \citenamefont {Machado}, \citenamefont {Chuang}, \citenamefont {Huang},\ and\
  \citenamefont {Choi}}]{allen2025quantum}%
  \BibitemOpen
  \bibfield  {author} {\bibinfo {author} {\bibfnamefont {R.~R.}\ \bibnamefont
  {Allen}}, \bibinfo {author} {\bibfnamefont {F.}~\bibnamefont {Machado}},
  \bibinfo {author} {\bibfnamefont {I.~L.}\ \bibnamefont {Chuang}}, \bibinfo
  {author} {\bibfnamefont {H.-Y.}\ \bibnamefont {Huang}},\ and\ \bibinfo
  {author} {\bibfnamefont {S.}~\bibnamefont {Choi}},\ }\bibfield  {title}
  {\bibinfo {title} {Quantum computing enhanced sensing},\ }\href
  {https://arxiv.org/abs/2501.07625} {\bibfield  {journal} {\bibinfo  {journal}
  {arXiv preprint arXiv:2501.07625}\ } (\bibinfo {year} {2025})}\BibitemShut
  {NoStop}%
\bibitem [{\citenamefont {Bennett}\ \emph {et~al.}(1997)\citenamefont
  {Bennett}, \citenamefont {Bernstein}, \citenamefont {Brassard},\ and\
  \citenamefont {Vazirani}}]{bennett1997strengths}%
  \BibitemOpen
  \bibfield  {author} {\bibinfo {author} {\bibfnamefont {C.~H.}\ \bibnamefont
  {Bennett}}, \bibinfo {author} {\bibfnamefont {E.}~\bibnamefont {Bernstein}},
  \bibinfo {author} {\bibfnamefont {G.}~\bibnamefont {Brassard}},\ and\
  \bibinfo {author} {\bibfnamefont {U.}~\bibnamefont {Vazirani}},\ }\bibfield
  {title} {\bibinfo {title} {Strengths and weaknesses of quantum computing},\
  }\href {https://doi.org/10.1137/S0097539796300933} {\bibfield  {journal}
  {\bibinfo  {journal} {SIAM Journal on Computing}\ }\textbf {\bibinfo {volume}
  {26}},\ \bibinfo {pages} {1510} (\bibinfo {year} {1997})}\BibitemShut
  {NoStop}%
\bibitem [{Note1()}]{Note1}%
  \BibitemOpen
  \bibinfo {note} {Note that the proof in~\protect \cref {app: grover_proof}
  also bounds the success probability for general state-discrimination, which
  upper bounds~\protect \cref {eq:P_succ}.}\BibitemShut {Stop}%
\bibitem [{\citenamefont {Leng}\ \emph {et~al.}(2023)\citenamefont {Leng},
  \citenamefont {Yang},\ and\ \citenamefont {Wang}}]{leng2023quantum}%
  \BibitemOpen
  \bibfield  {author} {\bibinfo {author} {\bibfnamefont {J.}~\bibnamefont
  {Leng}}, \bibinfo {author} {\bibfnamefont {F.}~\bibnamefont {Yang}},\ and\
  \bibinfo {author} {\bibfnamefont {X.-B.}\ \bibnamefont {Wang}},\ }\bibfield
  {title} {\bibinfo {title} {Quantum advantage of noisy grover's algorithm},\
  }\href {https://arxiv.org/abs/2306.10855} {\bibfield  {journal} {\bibinfo
  {journal} {arXiv preprint arXiv:2306.10855}\ } (\bibinfo {year}
  {2023})}\BibitemShut {NoStop}%
\bibitem [{\citenamefont {Zhou}\ \emph {et~al.}(2018)\citenamefont {Zhou},
  \citenamefont {Zhang}, \citenamefont {Preskill},\ and\ \citenamefont
  {Jiang}}]{zhou2018achieving}%
  \BibitemOpen
  \bibfield  {author} {\bibinfo {author} {\bibfnamefont {S.}~\bibnamefont
  {Zhou}}, \bibinfo {author} {\bibfnamefont {M.}~\bibnamefont {Zhang}},
  \bibinfo {author} {\bibfnamefont {J.}~\bibnamefont {Preskill}},\ and\
  \bibinfo {author} {\bibfnamefont {L.}~\bibnamefont {Jiang}},\ }\bibfield
  {title} {\bibinfo {title} {Achieving the heisenberg limit in quantum
  metrology using quantum error correction},\ }\href
  {https://doi.org/10.1038/s41467-017-02510-3} {\bibfield  {journal} {\bibinfo
  {journal} {Nature Communications}\ }\textbf {\bibinfo {volume} {9}},\
  \bibinfo {pages} {78} (\bibinfo {year} {2018})}\BibitemShut {NoStop}%
\bibitem [{\citenamefont {Conlon}\ \emph {et~al.}(2026)\citenamefont {Conlon},
  \citenamefont {Wang}, \citenamefont {Abbasgholinejad}, \citenamefont
  {Albert}, \citenamefont {Gullans},\ and\ \citenamefont
  {Gorshkov}}]{conlon2026fault}%
  \BibitemOpen
  \bibfield  {author} {\bibinfo {author} {\bibfnamefont {L.~O.}\ \bibnamefont
  {Conlon}}, \bibinfo {author} {\bibfnamefont {Y.-X.}\ \bibnamefont {Wang}},
  \bibinfo {author} {\bibfnamefont {E.}~\bibnamefont {Abbasgholinejad}},
  \bibinfo {author} {\bibfnamefont {V.~V.}\ \bibnamefont {Albert}}, \bibinfo
  {author} {\bibfnamefont {M.~J.}\ \bibnamefont {Gullans}},\ and\ \bibinfo
  {author} {\bibfnamefont {A.~V.}\ \bibnamefont {Gorshkov}},\ }\bibfield
  {title} {\bibinfo {title} {Fault-tolerant heisenberg-limited quantum
  sensing},\ }\href {https://arxiv.org/abs/2608.00171} {\bibfield  {journal}
  {\bibinfo  {journal} {arXiv preprint arXiv:2608.00171}\ } (\bibinfo {year}
  {2026})}\BibitemShut {NoStop}%
\bibitem [{\citenamefont {Demkowicz-Dobrza{\'n}ski}\ \emph
  {et~al.}(2017)\citenamefont {Demkowicz-Dobrza{\'n}ski}, \citenamefont
  {Czajkowski},\ and\ \citenamefont {Sekatski}}]{demkowicz2017adaptive}%
  \BibitemOpen
  \bibfield  {author} {\bibinfo {author} {\bibfnamefont {R.}~\bibnamefont
  {Demkowicz-Dobrza{\'n}ski}}, \bibinfo {author} {\bibfnamefont
  {J.}~\bibnamefont {Czajkowski}},\ and\ \bibinfo {author} {\bibfnamefont
  {P.}~\bibnamefont {Sekatski}},\ }\bibfield  {title} {\bibinfo {title}
  {Adaptive quantum metrology under general markovian noise},\ }\href
  {https://doi.org/10.1103/PhysRevX.7.041009} {\bibfield  {journal} {\bibinfo
  {journal} {Physical Review X}\ }\textbf {\bibinfo {volume} {7}},\ \bibinfo
  {pages} {041009} (\bibinfo {year} {2017})}\BibitemShut {NoStop}%
\bibitem [{\citenamefont {Helstrom}(1976)}]{Helstrom1976}%
  \BibitemOpen
  \bibfield  {author} {\bibinfo {author} {\bibfnamefont {C.~W.}\ \bibnamefont
  {Helstrom}},\ }\href {https://books.google.com/books?id=fv9SAAAAMAAJ} {\emph
  {\bibinfo {title} {Quantum Detection and Estimation Theory}}}\ (\bibinfo
  {publisher} {Academic Press},\ \bibinfo {address} {New York},\ \bibinfo
  {year} {1976})\BibitemShut {NoStop}%
\bibitem [{\citenamefont {Holevo}(1973)}]{Holevo1973}%
  \BibitemOpen
  \bibfield  {author} {\bibinfo {author} {\bibfnamefont {A.~S.}\ \bibnamefont
  {Holevo}},\ }\bibfield  {title} {\bibinfo {title} {Statistical decision
  theory for quantum systems},\ }\href
  {https://doi.org/10.1016/0047-259X(73)90028-6} {\bibfield  {journal}
  {\bibinfo  {journal} {Journal of Multivariate Analysis}\ }\textbf {\bibinfo
  {volume} {3}},\ \bibinfo {pages} {337} (\bibinfo {year} {1973})}\BibitemShut
  {NoStop}%
\bibitem [{\citenamefont {Barnett}\ and\ \citenamefont
  {Croke}(2009)}]{BarnettCroke2009}%
  \BibitemOpen
  \bibfield  {author} {\bibinfo {author} {\bibfnamefont {S.~M.}\ \bibnamefont
  {Barnett}}\ and\ \bibinfo {author} {\bibfnamefont {S.}~\bibnamefont
  {Croke}},\ }\bibfield  {title} {\bibinfo {title} {Quantum state
  discrimination},\ }\href {https://doi.org/10.1364/AOP.1.000238} {\bibfield
  {journal} {\bibinfo  {journal} {Advances in Optics and Photonics}\ }\textbf
  {\bibinfo {volume} {1}},\ \bibinfo {pages} {238} (\bibinfo {year}
  {2009})}\BibitemShut {NoStop}%
\bibitem [{\citenamefont {Hayashi}\ and\ \citenamefont
  {Nagaoka}(2003)}]{HayashiNagaoka2003}%
  \BibitemOpen
  \bibfield  {author} {\bibinfo {author} {\bibfnamefont {M.}~\bibnamefont
  {Hayashi}}\ and\ \bibinfo {author} {\bibfnamefont {H.}~\bibnamefont
  {Nagaoka}},\ }\bibfield  {title} {\bibinfo {title} {General formulas for
  capacity of classical-quantum channels},\ }\href
  {https://doi.org/10.1109/TIT.2003.813556} {\bibfield  {journal} {\bibinfo
  {journal} {IEEE Transactions on Information Theory}\ }\textbf {\bibinfo
  {volume} {49}},\ \bibinfo {pages} {1753} (\bibinfo {year}
  {2003})}\BibitemShut {NoStop}%
\end{thebibliography}
\end{document}